\documentclass[11pt]{article}
\usepackage[a4paper, total={16cm, 24cm}]{geometry}

\usepackage{natbib}
\renewcommand\tableofcontents{\listoftoc*{toc}} % no "Contents" header

\usepackage{authblk}
\author[1]{Haris Aziz}
\author[2]{Jiarui Gan}
\author[3]{Grzegorz Lisowski}
\author[1]{Ali Pourmiri}

\affil[1]{UNSW Sydney, Australia} 
\affil[2]{University of Oxford, United Kingdom} 
\affil[3]{University of Groningen, The Netherlands} 
\date{}

\usepackage{amsmath}
\usepackage{array,xspace,multirow,hhline,graphicx,xcolor,tikz,colortbl,tabularx,amsmath,amssymb,amsfonts}
\usetikzlibrary{shapes}
\usetikzlibrary{arrows}
\usetikzlibrary{decorations.pathreplacing,calligraphy}

\usepackage{caption,tabularx,booktabs}

\usepackage{graphicx}
\usepackage{caption}
\usepackage{subcaption}
\usepackage{collcell}
\usepackage{hhline}
\usepackage{pgf}
\usepackage{pgfplots}

\usepackage{booktabs}

\usepackage{wrapfig}
\usepackage{amsthm}
\newcommand\eat[1]{}
\usepackage{pgfplots}
\usepackage{bbm}
\usepackage{verbatim,ifthen}
\usepackage{pifont}
\usepackage{ifthen}
\usepackage{pgflibraryshapes}

\usepackage{nicefrac}
\usetikzlibrary{arrows}
\usepackage{ltxtable}
\usepackage{longtable}
\usepackage[normalem]{ulem}
	\usepackage{varioref}

\usepackage{calc}
\newsavebox\CBox
\newcommand\hcancel[2][0.5pt]{%
  \ifmmode\sbox\CBox{$#2$}\else\sbox\CBox{#2}\fi%
  \makebox[0pt][l]{\usebox\CBox}%  
  \rule[0.5\ht\CBox-#1/2]{\wd\CBox}{#1}}

\tikzset{
  jumpdot/.style={mark=*,solid},
  excl/.append style={jumpdot,fill=white},
  incl/.append style={jumpdot,fill=black},
  rexcl/.append style={jumpdot,color=red,fill=white},
  rincl/.append style={jumpdot,fill=black,color=red},
}
\newcommand{\stkout}[1]{\ifmmode\text{\sout{\ensuremath{#1}}}\else\sout{#1}\fi}

\usepackage[ruled]{algorithm2e} % For algorithms

\SetAlFnt{\small}
\SetAlCapFnt{\small}
\SetAlCapNameFnt{\small}
\SetAlCapHSkip{0pt}
\IncMargin{-\parindent}

\usepackage{amsmath}
\usepackage{mathtools}

\DeclarePairedDelimiter\floor{\lfloor}{\rfloor}

		\newcommand{\haris}[1]{{\textcolor{black}{Haris says: #1}}}
				\newcommand{\harisnew}[1]{{\textcolor{black}{#1}}}

\definecolor{gray(x11gray)}{rgb}{0.75, 0.75, 0.75}
\def\colorModel{rgb} %You can use rgb or hsb

\newcommand\ColCell[1]{
  \pgfmathparse{#1<0.5?1:0}  %Threshold for changing the font color into the cells
    \ifnum\pgfmathresult=0\relax\color{white}\fi
  \pgfmathsetmacro\compA{1-#1}      %Component R or H
  \pgfmathsetmacro\compB{1-#1/1.5} %Component G or S
  \pgfmathsetmacro\compC{1}      %Component B or B
  \edef\x{\noexpand\centering\noexpand\cellcolor[\colorModel]{\compA,\compB,\compC}}\x #1
  } 
\newcolumntype{E}{>{\collectcell\ColCell}m{0.5cm}<{\endcollectcell}}  
	\usepackage{boxedminipage}
			\usepackage{xspace}

\newcommand{\ali}[1]{\textcolor{red}{\textbf{Ali says:} #1}}

\newcommand{\greg}[1]{\textcolor{red}{\textbf{Grzegorz says:} #1}}

\newcommand{\gregR}[1]{\textcolor{black}{#1}}
\newcommand{\gregRN}[1]{\textcolor{black}{#1}}

	\renewcommand{\Pr}[1]{\ensuremath{\operatorname{\mathbf{Pr}}\left[#1\right]}}

	\newcommand{\Var}[1]{\ensuremath{\operatorname{\mathbf{Var}}\left[#1\right]}}

	\newcommand{\Ex}[1]{\ensuremath{\operatorname{\mathbf{E}}\left[#1\right]}}

			\newcommand{\pbDef}[3]{%
			\smallskip
			\noindent
			\begin{center}
			\begin{boxedminipage}{0.98 \columnwidth}
			#1\\[5pt]
			\begin{tabular}{l p{0.75 \columnwidth}}
			Input: & #2\\
			Question: & #3
			\end{tabular}
			\end{boxedminipage}
			\end{center}
			\smallskip
			}

	\usepackage{xr}
		
\newcommand{\G}{\mathcal{G}}

\newtheorem{theorem}{Theorem}
\newtheorem{proposition}{Proposition}
\newtheorem{corollary}{Corollary}
\newtheorem{lemma}{Lemma}
\newtheorem{example}{Example}

\begin{document}
		\title{The Team Order Problem: Maximizing the Probability of Matching Being Large Enough }

%\titlerunning{The Team Order Problem}

\maketitle

\begin{abstract}
We consider a matching problem, which is meaningful in team competitions, as well as in information theory, recommender systems, and assignment problems. In the competitions which we study, each competitor in a team order plays a match with the corresponding opposing player. The team that wins more matches wins.  We consider a problem where the input is the graph of probabilities that a team 1 player can win against the team 2 player, and the output is the optimal ordering of team 1 players given the fixed ordering of team 2. Our central result is a polynomial-time approximation scheme (PTAS) to compute a matching whose winning probability is at most $\varepsilon$ less than the winning probability of the optimal matching. We also provide tractability results for several special cases of the problem, as well as an analytical bound on how far the winning probability of a maximum weight matching of the underlying graph is from the best achievable winning probability.
 
\end{abstract}

\section{Introduction}

Bipartite matching underpins several impactful problems in allocation and market design problems including kidney allocation, adword auctions, on demand taxi allocation, refugee assignment, or school choice~(see, e.g., \cite{EIV23a}).  
We consider a fundamental matching problem with an underlying weighted bipartite graph where each edge weight has weight between 0 and 1. Instead of focusing on the classical objective of maximizing the total weight of the matching, we focus on a different objective with a probabilistic interpretation: We want to compute a matching that maximizes the probability of reaching a target size. This problem 
can model several scenarios, including that of the so called \emph{team order problem}.
		
\gregR{One of the most relevant applications of our setting is the rivalry between teams of contestants. }Consider team competitions in which both teams put forward an ordering of \gregR{their} players. The contestants then play matches against the corresponding contestants from the opposing team. The team that wins more matches wins the overall competition. 
Such competitions are not only held in various inter-club tennis competitions, the same format is also used in international table tennis and badminton competitions, such as the Corbillon Cup, Swaythling Cup, Thomas Cup, and the Olympic Games. 
\harisnew{We focus on the problem in which one team's order is fixed (as is the case in many situations where the home team commits to an ordering) and the other team wants to compute the optimal ordering. As the ordering of one team is fixed, the problem of computing the other team's ordering is essentially a competitor matching problem. }

The problem of finding a way to maximize the number of achieved goals by setting an appropriate line-up is not limited to sport competitions. \harisnew{Indeed, it admits several other motivations in competitive contexts such as politics (fielding political candidates in different constituencies against candidates of a rival party).}
Our problem also provides a perspective into finding \emph{durable matchings}. Suppose that we are given the probability of success of various partnerships. 
For example a partnership could represent a job placement or allocation of refugee family to a council (see, e.g., \cite{BFH+18a,AGP+21a}). A typical objective could be maximizing the expected number of partnerships. However, another meaningful objective that is centred around a particular target could be to maximize the probability of having a target number of successful partnerships\gregR{, which maps to the objective that we study}. 
Another \gregR{potential application of our research relates to \emph{information networks} (see, e.g., \cite{li2021inferring}). Suppose that we are given such a network,} represented by a flow network. \gregR{There, each} edge has a reliability probability of a message reaching the other side\gregR{, and we} want to find a flow maximizing probability of delivering a target number of messages. Finally, \gregR{our research is motivated by its applications in \emph{recommendation systems} (see, e.g., \cite{mohamed2019recommender}).} Suppose that a ranked list of recommendations needs to be displayed with each item having a probability of being clicked depending on its position in the ranking list. One may want to maximize the probability of having a target number of items being clicked\gregR{, which can be captured by our problem.} 
  		We explore the following questions. 

  		\begin{quote}
  		\emph{How hard is the team order problem? Under what conditions is it easy to solve? What are reasonable approximation approaches for the problem?}
  		\end{quote}

We note that the problem that we study in this paper is closely related to the maximum-weight matching problem. There, we are given a bipartite graph, where each edge is assigned a weight, and the objective is to find a matching with the maximum sum of weights. \gregRN{In fact, our results reflect that finding the solution to that problem provides a good approximation of the optimal solution.} However, the problem we study is substantially more complex. Indeed, for an instance of the team order problem to be positive, we require that the weights in a selected matching are large enough for some subset of edges, instead of maximizing their global sum. \gregRN{Furthermore, given the strategic games interpretation of our setting, our results concern the computation of the optimal response to the opponent choice, which is an important step towards the study of equilibria in this setting.}

  		\paragraph{Contributions.}

We first show that the winning probability of a given matching (line-up) can be computed in polynomial time \gregR{(Proposition \ref{pro:winning})}.
\gregR{Subsequently, we show that in certain settings computing an optimal line-up is tractable. In particular, when} the input winning probability of each partnership takes its value from a size-three set $\{\alpha, \beta, 0\}$ 
we show that the optimal matching can be computed in polynomial time \gregR{(Theorem \ref{thm:ab0})}. \gregRN{While we conjecture that the team order problem is hard in the general case, we show that it is tractable for practical purposes.}
\gregR{Our central result} is a \emph{polynomial-time approximation scheme (PTAS)}\footnote{\gregR{A PTAS is a scheme which, for every instance of a problem and $\varepsilon>0$, provides an approximate solution based on $\varepsilon$.}} to compute a matching whose winning probability is at most \gregR{$\varepsilon$} less than the winning probability of the optimal matching \gregR{(Theorem \ref{thm:alg2})}. Although the winning probability is not a linear objective, we show that the general problem of computing an optimal matching can be solved via integer linear programming.  \gregR{Also, we provide an analytical bound on how far the winning probability of a maximum weight matching is from the best achievable winning probability.}
%The full proofs of all the statements are \gregR{provided} in the appendix. 

  	\section{Related Work}
	\gregR{Our results are relevant to a number of research direction in multi-agent systems.}

 \paragraph{\gregR{Matching Theory.}}
  Matching problems have been widely studied in combinatorial optimization. The standard objectives typically focus on maximizing the weight of the matching~(see, e.g., \cite{BAM09a,LoPl09a}). In our context, maximizing the weight of the underlying weighted bipartite graph gives us a matching maximizing the expected number of matches won. Our objective is different as we want to maximize the probability of winning a target number of matches. 
  The paper most relevant to our work is by Tang et al.~\cite{TSL09a}, which concerns
  the same setting but considered different problems. It takes an economic design approach  and presents necessary and sufficient conditions, ensuring that truthful reporting and maximal effort in matches are equilibrium strategies. 
  \gregRN{We note that the probabilistic approach in matching has been previously studied. E.g., Aziz et al. \cite{aziz2020stable} studied the stable matching problem with uncertain preferences.}

 \paragraph{\gregR{Manipulation of Competitions.}}
  Within the wider topic of manipulations in competitions, there have been several papers on identifying conditions or manipulations under which a certain team or player can win. A notable example is manipulating the draw of a balanced knockout tournament to maximize the probability of a certain player winning, i.e., the \emph{tournament fixing problem}~\cite{Will15a,AGM+14a,vu2009complexity}. %In sports knockout tournaments, it is very unlikely that someone can unilaterally rig the draw in someone's favor. On the other hand, strategizing the contestant ordering in weekend team competitions is more plausible. 
  Similarly, there has also been algorithmic research on round-robin formats to understand which teams have a chance to win the overall tournament~\cite{KePa01a,ABF+15a}.
	
\paragraph{\gregR{Colonel Blotto Game.}}  	\gregR{Furthermore, } the team line-up setting bears resemblance to Colonel Blotto Games which are two-player zero-sum games in which two armies fight in $n$ battle fields with each battle being won by the army that had more troops in the battle (see, e.g., \cite{Robe06b,ShWe78a}). The armies are interested in  maximizing a weighted sum of utilities from the battlefields where they gain victories. Although the team-line-up setting is similar in that each battle corresponds to a match, in Colonel Blotto games, the armies have more flexibility in shuffling their troops around. Secondly, in Colonel Blotto games the outcome of a battle depends on the \emph{number} of troops of each army whereas in the team line up setting, the outcome of a match depends on the identities of the respective players.  \harisnew{Independent of our work, Gaonkar et al.~\citep{GRW22} considered a version of Blotto games in which every resource is unique and non-interchangeable which makes it close to our setting. They motivate the problem as \textit{derby games} in which teams assign each resource to a particular round and wins a payoff corresponding to that round if they win the round.} \gregR{ We note, however, that our work differs significantly from their results. In particular, they examine Nash equilibria, which are not the focus of our study. Furthermore, they do not take the information on winning probabilities into account and do not focus on algorithmic issues.}

\paragraph{\gregR{Sequential Games.}} \gregR{ Games between teams of players in which the ordering of contestants matters gained a substantial interest in recent literature. Fu et al. \cite{10.1257/aer.20121469} studied the scenario in which teams compete in a number of games between pairs of players.  Within this setting they investigated how the sequencing of those matches impacts the result. We note that, in contrast to our study, the games they considered are also based on private rewards for the individual players. 
Furthermore, Konishi et al. \cite{KONISHI2022274} studied the problem of whether the equilibrium winning probability in such games depends on whether matches are held simultaneously, or sequentially. Also, Fu and Lu \cite{fu2020equilibrium} explored the topic of how teams can strategically assign contestants to time-slots of a sequential competition. % depending on their strengths. 
Let us further note that in contrast to our work the discussed papers on sequential games do not focus on computational complexity.}

\paragraph{\gregR{Nominee Selection.}}
\gregR{Our setting is also related to the literature on  strategic selection of group members participating in a competition. In social choice theory, this problem relates to the process of selecting representative for the elections (see, e.g., \cite{faliszewski2016hard,DBLP:conf/aaai/borodin2019primarily}). Regarding sport competitions, our problem relates to choosing a coalition member to participate in a tournament (see, e.g.,  \cite{aamastours,10.1007/978-3-031-20614-6_14}).}

  		\section{The Team Order Problem}
		
  	We consider the following problem setting.

  		\begin{itemize}
  			\item % Set of teams $T=\{T_1,\ldots, T_m\}$
  			Two teams $T_1$ and $T_2$ are to play a team competition.
  			\item Each team $T_i$ has $n$ contestants $t_i^1,\dots, t_i^n$.
  			\item We have information about the winning probability $p(t_i^a, t_j^b)$ of any contestant $t_i^a$ against any other contestant $t_j^b$. The instance is said to be \emph{degenerate} if all the winning probabilities are 0 or 1. 
  		\end{itemize}

		In the competition each team is required to report a line-up, i.e., an ordering $i_1, \dots, i_n$ of its contestants, which is a permutation of $1,\dots,n$. Then each contestant $t_i^{i_k}$ plays a match with the corresponding contestant $t_j^{j_k}$. 
		The team that wins at least $\floor{\frac{n}{2}}+1$ matches wins the competition. 
		All of our results hold equally well if the target $\floor{\frac{n}{2}}+1$ is replaced by some generic target $L$ that is higher or lower than $\floor{\frac{n}{2}}+1$.

  		We will consider computational problems related to strategic aspects of deciding on a line-up of players of a team. 
		Our primary consideration is the following problem of computing the best response to a given line-up of the opposing team.

  		 \pbDef{\textsc{Team Order}}{A target probability $q\in [0,1]$ and a finite set Team Order instance, and a (deterministic) line-up of team $T_2$.}
  		{Does there exists a line-up for team $T_1$ under which the probability of $T_1$ winning against $T_2$ is at least $q$?}
								 
		 Without loss of generality, we can assume that the line-up of $T_2$ is fixed to $t_2^1,\dots, t_2^n$ when dealing with the \textsc{Team Order} problem.
		From a graph theoretic perspective, it can be captured by a weighted and complete bipartite graph $G=(T_1\cup T_2,E,p)$. The weight of an edge $(t_i^a, t_j^b)$ is winning probability $p(t_i^a, t_j^b)$ of any contestant $t_i^a$ against any other contestant $t_j^b$. We will call $G$ the \gregR{\emph{corresponding graph}}. 		
		The line-ups of the two teams correspond to a perfect matching in $G$, which pairs up every player in $T_1$ with a unique player in $T_2$.
		 \gregR{Assuming that matches are independent,} we are interested in computing a perfect matching $M$ whose edge weights maximize the winning probability: 
		 \[
		 	\sum_{\stackrel{S\subseteq \{1,..,n\}}{|S|\ge \lfloor \frac{n}{2}\rfloor+1}} \prod_{i\in S}^n p(t_1^i, t_2^{M(i)}) \prod_{i \notin S}^n \left(1 - p(t_1^i, t_2^{M(i)}) \right),
		 \]
		 where $M(i)$ denotes the index of the player in $T_2$ who is \gregR{matched} with $t_1^i$, and each $S$ is an outcome of the competition represented as the set of players in $T_1$ who win against their opponents.
		 For simplicity, we will also write the probabilities as $p_{i,j} = p(t_1^i, t_2^j)$.
		 					
In fact, even when the line-ups of both teams are given, it is not immediately clear that the winning probability of $M$ can be computed efficiently, since there are exponentially (in $n$) many possible outcomes of the competition.
One way that leads to a polynomial-time algorithm to compute this probability is via dynamic programming, which results in the proposition below.

  \begin{proposition}\label{pro:winning}
	Given the line-ups of $T_1$ and $T_2$, the winning probability of each team can be computed in time $O(n^2)$.
  	\end{proposition}
	\begin{comment}
  	\begin{proof}
  The algorithm is via dynamic programming. Let $p_j$ be the probability of the first team winning the \gregR{$j^{\text{th}}$} match according to the line-ups. 
  We keep track of an array $A$, where $A[i,j]$ denotes the probability of $i$ wins after the first $j$ matches. 
  In such an array, $A[0,1]=1-p_1$ and \harisnew{$A[0,0]=1$, $A[i,0]=$0 for $i \neq 0$}, 
\harisnew{and $A[0,j] = \prod_{k=1}^{j} (1-p_k)$ for all $j \in [n].$}
 Furthermore, we can compute a given entry  $A[i,j]$  by $A[i,j]=A[i-1,j-1] \cdot p_j+A[i,j-1] \cdot (1-p_j)$. 
  The array takes size $n(n+1)$ so the probability of first team winning a given number of matches $k$ can be computed by reading off the entry $A[k,n]$. In order to compute the probability of the first team winning more matches, we simply need to find the sum $\sum_{k\geq \floor{\frac{n}{2}}+1}A[k,n].$

   \end{proof}
\end{comment}
   
We present an example below to illustrate the problem.

     						 		\begin{example}
   								\label{exp:max-weight}
								
									\gregR{Take an instance with the input winning probabilities as in Table~\ref{tb:1}.
									Also, Team $T_1$ has $3!$ different line-ups $O_1, \dots, O_6$ as illustrated in Figure~\ref{fig:1}.}

             	\begin{table}[h]
					 		\renewcommand{\arraystretch}{1.1}
					 		\begin{center}
					 	%	\small
					 		\begin{tabular}{ c| c c  c }
					 					& $t_2^1$	& $t_2^2$	& $t_2^3$	\\
					 		\hline %------------------------------------
					 		$t_1^1$		& $0.9^*$		& $\underline{1}$ 		& $1$ 		\\
					 		$t_1^2$		& $0.5$		& $0.9^*$		& $\underline{1}$ 		\\
					 		$t_1^3$		& $\underline{0}$ 		& $0.5$		& $0.9^*$		\\
					 		%\bottomrule
					 		\end{tabular}
					 		\end{center}
				 		
							\caption{Each entry $(i,j)$ is the probability $p(t_{1}^i, t_2^j)$.}
						\label{tb:1}
					 		% \caption{An example showing that {\sc MaxPF} and {\sc MaxEM} are different, where each entry $(a,b)$ is the probability $p(t_{2}^a, t_1^b)$.}
					 		\end{table}

   \begin{figure}[t]
   \begin{center}
   		\scalebox{.9}{\hspace{4em}{	\begin{tikzpicture}
   	\draw
   	(1, 1) node[circle, black, draw, inner sep=.15](1-1){$t^3_1$}
   	(1, 1.8) node[circle, black, draw, inner sep=.15](1-2){$t^2_1$}
   	(1, 2.6) node[circle, black, draw, inner sep=.15](1-3){$t^1_1$}
   	(2.7, 1) node[circle, black, draw, inner sep=.15](2-1){$t^3_2$}
   	(2.7, 1.8) node[circle, black, draw, inner sep=.15](2-2){$t^2_2$}
   	(2.7, 2.6) node[circle, black, draw, inner sep=.15](2-3){$t^1_2$}
   	(4, 1) node[circle, black, draw, inner sep=.15](3-1){$t^3_1$}
   	(4, 1.8) node[circle, black, draw, inner sep=.15](3-2){$t^2_1$}
   	(4, 2.6) node[circle, black, draw, inner sep=.15](3-3){$t^1_1$}
   	(5.7, 1) node[circle, black, draw, inner sep=.15](4-1){$t^3_2$}
   	(5.7, 1.8) node[circle, black, draw, inner sep=.15](4-2){$t^2_2$}
   	(5.7, 2.6) node[circle, black, draw, inner sep=.15](4-3){$t^1_2$}
   	(7, 1) node[circle, black, draw, inner sep=.15](5-1){$t^3_1$}
   	(7, 1.8) node[circle, black, draw, inner sep=.15](5-2){$t^2_1$}
   	(7, 2.6) node[circle, black, draw, inner sep=.15](5-3){$t^1_1$}
   	(8.7, 1) node[circle, black, draw, inner sep=.15](6-1){$t^3_2$}
   	(8.7, 1.8) node[circle, black, draw, inner sep=.15](6-2){$t^2_2$}
   	(8.7, 2.6) node[circle, black, draw, inner sep=.15](6-3){$t^1_2$}
   	(1, 5) node[circle, black, draw, inner sep=.15](7-1){$t^3_1$}
   	(1, 5.8) node[circle, black, draw, inner sep=.15](7-2){$t^2_1$}
   	(1, 6.6) node[circle, black, draw, inner sep=.15](7-3){$t^1_1$}
   	(2.7, 5) node[circle, black, draw, inner sep=.15](8-1){$t^3_2$}
   	(2.7, 5.8) node[circle, black, draw, inner sep=.15](8-2){$t^2_2$}
   	(2.7, 6.6) node[circle, black, draw, inner sep=.15](8-3){$t^1_2$}
   	(4, 5) node[circle, black, draw, inner sep=.15](9-1){$t^3_1$}
   	(4, 5.8) node[circle, black, draw, inner sep=.15](9-2){$t^2_1$}
   	(4, 6.6) node[circle, black, draw, inner sep=.15](9-3){$t^1_1$}
   	(5.7, 5) node[circle, black, draw, inner sep=.15](10-1){$t^3_2$}
   	(5.7, 5.8) node[circle, black, draw, inner sep=.15](10-2){$t^2_2$}
   	(5.7, 6.6) node[circle, black, draw, inner sep=.15](10-3){$t^1_2$}
   	(7, 5) node[circle, black, draw, inner sep=.15](11-1){$t^3_1$}
   	(7, 5.8) node[circle, black, draw, inner sep=.15](11-2){$t^2_1$}
   	(7, 6.6) node[circle, black, draw, inner sep=.15](11-3){$t^1_1$}
   	(8.7, 5) node[circle, black, draw, inner sep=.15](12-1){$t^3_2$}
   	(8.7, 5.8) node[circle, black, draw, inner sep=.15](12-2){$t^2_2$}
   	(8.7, 6.6) node[circle, black, draw, inner sep=.15](12-3){$t^1_2$};
	
   	\draw [line width=.25mm, black] (1-1) -- node[below, pos=.2] {$0$} (2-3);
   	\draw [line width=.25mm, black] (1-2) -- node[above, pos=0.9 ] {$.9$} (2-2);
   	\draw [line width=.25mm, black] (1-3) -- node[above, pos=.2] {$1$} (2-1);
   	\draw [line width=.25mm, black] (3-1) -- node[below, pos=.4] {$0.5$} (4-2);
   	\draw [line width=.25mm, black] (3-2) -- node[above,pos=.5] {$0.5$} (4-3);
   	\draw [line width=.25mm, black] (3-3) -- node[above,pos =.6] {$1$} (4-1);
   	\draw [line width=.25mm, black] (5-1) -- node[above,pos = .4] {$0$} (6-3);
   	\draw [line width=.25mm, black] (5-2) -- node[below, pos=.4] {$1$} (6-1);
   	\draw [line width=.25mm, black] (5-3) -- node[above, pos=.3] {$1$} (6-2);
   	\draw [line width=.25mm, black] (7-1) -- node[above] {$0.9$} (8-1);
   	\draw [line width=.25mm, black] (7-2) -- node[above] {$0.9$} (8-2);
   	\draw [line width=.25mm, black] (7-3) -- node[above] {$0.9$} (8-3);
   	\draw [line width=.25mm, black] (9-1) -- node[below, pos=.3] {$0.5$} (10-2);
   	\draw [line width=.25mm, black] (9-2) -- node[above, pos=.2] {$1$} (10-1);
   	\draw [line width=.25mm, black] (9-3) -- node[above] {$0.9$} (10-3);
   	\draw [line width=.25mm, black] (11-1) -- node[above] {$0.9$} (12-1);
   	\draw [line width=.25mm, black] (11-2) -- node[below, pos=.3] {$0.5$} (12-3);
   	\draw [line width=.25mm, black] (11-3) -- node[above, pos=.2] {$1$} (12-2);
	
   	\node[text width = 4cm] at (2.8,4) {$O_1=( t^1_1, t^2_1, t^3_1)$};
   	\node[text width = 4cm] at (5.75,4) {$O_2=( t^1_1, t^3_1, t^2_1)$};
   	\node[text width = 4cm] at (8.85,4) {$O_3=( t^2_1, t^1_1, t^3_1)$};
   	\node[text width = 4cm] at (2.8,0) {$O_4=( t^3_1, t^2_1, t^2_1)$};
   	\node[text width = 4cm] at (5.75,0) {$O_5=( t^2_1, t^3_1, t^1_1)$};
   	\node[text width = 4cm] at (8.85,0) {$O_6=( t^3_1, t^1_1, t^2_1)$};
   	\end{tikzpicture}}}
    \end{center}
   \caption{Graph theoretic view of Example~\ref{exp:max-weight}. There are $3!$ different line-ups for $T_1$  and each line-up is a perfect matching and has its own winning probabilities illustrated on the edges.
   \label{fig:1}}
   \end{figure}
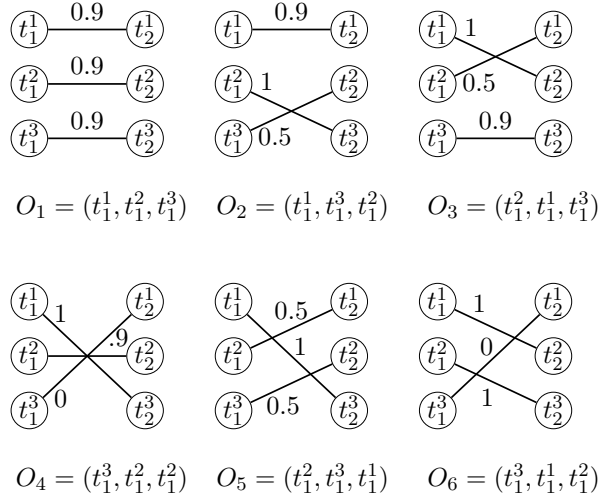
  
									Suppose that $T_2$ uses the line-up $(t_2^1, t_2^2, t_2^3)$. 
									If $T_1$ responds with $(t_1^3, t_1^1, t_1^2)$ (underlined entries), the probability that they beat $T_2$ is $1$, as they will win two matches with certainty. On the other hand, if $T_1$ responds with $(t_1^1, t_1^2, t_1^3)$ (starred entries), their winning probability becomes
									\[
  			\underbrace{0.9 \times 0.9 \times 0.9}_{\text{prob. of winning all the matches}}+ \underbrace{0.9 \times 0.9 \times (1-0.9) \times 3}_{\text{prob. of winning exactly two matches}} = 0.972.
  		\]

     						 \end{example}

Indeed, in the above example, the line-up $(t_1^1, t_1^2, t_1^3)$ corresponds to the perfect matching with the maximum total weight in this instance. This demonstrates that weight maximizing matchings may not be optimal solutions to \textsc{Team Order}. 
The next example shows that such matchings fail to even provide any approximation guarantee to \textsc{Team Order}.
		
		\begin{table}[h]
     						 			\renewcommand{\arraystretch}{1.2}
     						 			\begin{center}
     						 			%\small
     						 		\scalebox{1}{	\begin{tabular}{ c | c c  c ccccc }
     						 						& $t_2^1$	& $t_2^2$	& $t_2^3$ & $t_2^4$& $t_2^5$ & $t_2^6$ & $t_2^7$ 	\\
     						 		\hline
     						 		 %------------------------------------
     						 			% $t_1^1$	& $t_2^1$	& $t_2^2$	& $t_2^3$ & $t_2^4$& $t_2^5$ & $t_2^6$ & $t_2^7$ 	\\
     						 			$t_1^1$			& $0$	& $0$	& $0$ & $0.5$& $1$ & $1$ & $1$ 	\\
     						 			$t_1^2$		& $0$	& $0$	& $0$ & $0$& $0.5$ & $1$ & $1$ 	\\
     						 			$t_1^3$		& $0$	& $0$	& $0$ & $0$& $0$ & $0.5$ & $1$ 	\\
     						 			$t_1^4$		& $0$	& $0$	& $0$ & $0$& $0$ & $0$ & $0.5$ 	\\
     						 			$t_1^5$		& $0$	& $0$	& $0$ & $0$& $0$ & $0$ & $0$ 	\\
     						 			$t_1^6$		& $0$	& $0$	& $0$ & $0$& $0$ & $0$ & $0$
     						 			\\
     						 				$t_1^7$		& $0$	& $0$	& $0$ & $0$& $0$ & $0$ & $0$ 	\\
     						 			%\bottomrule
     						 			\end{tabular}}
     						 			\end{center}
  						 			
     						 			% \caption{An example showing that {\sc MaxPF} and {\sc MaxEM} are different, where each entry $(a,b)$ is the probability $p(t_{1}^a, t_2^b)$.}
     									\caption{Each entry $(i,j)$ is the probability $p(t_{1}^i, t_2^j)$.}
   									\label{tb:2}
     						 			\end{table}
									
     						 		\begin{example}
								\label{exp:max-weight-2}
   	Suppose that $n=7$ and the input winning probabilities are given in Table~\ref{tb:2}.
     						 			The maximum weight matching gives the guarantee of winning three matching with certainty but losing all the others, and hence probability $0$ of winning the competition. On the other hand, the matching that gives probability $0.5$ of winning four matches wins the competition with a non-zero probability. The example also shows that the maximum weight matching cannot approximate the highest winning probability within any multiplicative factor.

     						 			\end{example}
		
   								In the above example, the better solution has more balanced winning probabilities over the matches. In view of this, one may conjecture that a \emph{leximin-maximizing} matching is optimal for the \textsc{Team Order} problem.\footnote{\gregR{A vector $x$ is leximin-greater than a vector $y$ if $x$ and $y$ are in non-decreasing order and $x$ is lexicographically greater than $y$.}} However, the next example disproves this conjecture: a {leximin-maximizing} matching may not be optimal, even when it is also maximum weight matchings.

\begin{table}[h]
      							 			\renewcommand{\arraystretch}{1.2}
      							 			\begin{center}
      							 			%\small
      							 		\scalebox{1}{	\begin{tabular}{ c | c c  c ccccc }
      							 						& $t_2^1$	& $t_2^2$	& $t_2^3$ \\
      							 		\hline
      							 		 %------------------------------------
      							 			% $t_1^1$	& $t_2^1$	& $t_2^2$	& $t_2^3$ & $t_2^4$& $t_2^5$ & $t_2^6$ & $t_2^7$ 	\\
      							 			$t_1^1$			& $0.9$	& $0.5$&$1$\\
      							 			$t_1^2$	& $0.5$	& $0.1$&$1$\\
      										$t_1^3$	& $0$	& $0$&$1$\\
      							 			%\bottomrule
      							 			\end{tabular}}
      							 			\end{center}
      							 			\caption{Each entry $(i,j)$ is the probability $p(t_{1}^i, t_2^j)$.
      							 			\label{tb:third}}
      							 			\end{table}
										
      						 		\begin{example}
								\label{exp:leximin}
      						 			Suppose that $n=3$ and one match is guaranteed to be won as shown in Table~\ref{tb:third}.
    									The edge weights of the maximum weight matchings are (1) $0.5, 0.5, 1$, or (2) $0.1, 0.9, 1$, and the first one is a leximin-maximizing matching.
    %  						 			\begin{enumerate}
    %  						 				\item $0.5, 0.5, 1$ or
    %  						 				\item $0.1, 0.9, 1$
    %  						 			\end{enumerate}
      						 			However, the winning probabilities of these two matchings are $1-0.25=0.75$ and $1-0.09=0.89$, respectively.
    									%In that case there were two maximum size matchings but it is better go for the maximum size matching with lopsided weights.

      						 			\end{example}

   \section{Tractable Variants}

   \gregR{In this section we show that \textsc{Team Order}  
   is tractable if there are only two values of probabilities which are greater than 0 in an instance. Moreover, we demonstrate that checking if a team can win with a non-zero probability can be done in polynomial time. Finally, we show that finding the line-up maximizing winning all the matches is tractable.}
  \gregR{Our reasoning in this section is closely related to the \textsc{Maximum Weight Matching} problem. We note that it can be solved in $O(n^3)$ time via the Hungarian algorithm~\cite{kuhn55hungarian}}.

   \pbDef{\textsc{Maximum Weight Matching}}
		{A bipartite graph $G$, weight \gregR{$w(e) \in \mathbb{R}_+$} for each edge $e$ on $G$. }%\greg{$w(e) \in \mathbb{R}_+$?}}
		{Compute a perfect matching $M$ of $G$ that maximizes $w(M) :=\sum_{e\in M}w(e)$.}

\subsection{When Input Probabilities Have Three Values (Including 0)}

\gregR{Let us consider the case in which} the input probabilities are from a set $\{\alpha, \beta, 0\}$ and, without loss of generality, assume that {$\alpha > \beta > 0$}. We note that the problem appears closely connected to a \textsc{Colored Bipartite Matching} problem with two types of colors: given a bipartite graph with red and blue edges, does there exists a matching with (exactly) a certain number of red edges? Although the complexity of this red-blue matching problem is open~\cite{YMS02a}, we show that the optimal line-up problem can be solved in polynomial-time via Algorithm~\ref{alg1}.
We also remark that with this probability set $\{\alpha, \beta, 0\}$ the problem still remains different from \textsc{Maximum Weight Matching}, as we demonstrated via Example~\ref{exp:max-weight}.

\begin{algorithm}[t]
  	\caption{{\sc Iterative Algorithm} }\label{alg1}
  	%\begin{algorithmic}[1]
	
%	\BlankLine
	
  		\KwIn{a \textsc{Team Order} instance $G = (T_1 \cup T_2, E, p)$ where $p_{i,j} \in \{\alpha, \beta, 0\}$, $\alpha > \beta > 0$.}
		\KwOut{an optimal solution to \textsc{Team Order}.}
		
		\BlankLine
		
  		Remove all zero-weight edges of $G$;
		
  		 $opt \leftarrow 0$;
		 
  		\For{ $s=\lfloor \frac{n}{2} \rfloor+1,\ldots n$}{
  		 $M_s \leftarrow$ maximum weight matching of size $s$; \tcp{polynomial-time solvable}
  		\If {$M_s\neq \emptyset$}{
  		$p_s \leftarrow$ winning probability of line-up $M_s$; \tcp{see Proposition \ref{pro:winning}}
		
  		\If{$p_s>opt$}{
  			$opt \leftarrow p_s$;
			
  			$M^* \leftarrow M_s$;
  			}
  		}
  		}
		
		\Return $M^*$. 
		
  \end{algorithm}

\begin{theorem}\label{thm:ab0}
  		Suppose that $G=(T_1\cup T_2, E, p)$ is a \textsc{Team Order} instance with $p_{i,j} \in \{\alpha, \beta, 0\}$ for all $i,j \in \{1,\dots, n\}$. Then an optimal line-up can be computed in polynomial time.  
  		\end{theorem}
  	\begin{proof}
  		Suppose that $M^*$ denotes an optimal line-up. Let $X$ denote a random variable counting   the number of games won by $T_1$ corresponding  to $M^*$. Then $X$ follows a  Poisson Binomial (PB) distribution:
		\begin{equation*}
		\begin{split}
%\begin{aligned}
  	%	\[
  	X \sim PB(\underbrace{\alpha,\ldots,\alpha}_{x}, \underbrace{\beta,\ldots,\beta}_{y}, \underbrace{0,\ldots,0}_{z})=  \\ 
		PB(\underbrace{\alpha,\ldots,\alpha}_{x}, \underbrace{\beta,\ldots,\beta}_{y}),
  		%\]
		\end{split}
		\end{equation*}
  		where $x$, $y$ and $z$ are non-negative integers. 
  		Let us remove all $0$-weight edges from $G$ and call the resulting graph $G'$. Then, 
  	 $M^*$ is a matching of size  $x+y$  in $G'$. Also, any maximum weight matching of size $x+y$, say $M$, has at least $x$ $\alpha$-weight edges. Notice that if  $M$ has at least $x+1$ $\alpha$-weight edges, then Poisson binomial random variable $Y$ corresponding to $M$  stochastically dominates $X$ contradicting the fact that $M^*$ is an optimal line-up.
  	 The argument also suggests that searching \gregR{through}  all matchings of various sizes will hit the optimal line-up.
  	 Note that finding a maximum weight matching of a given size is polynomially solvable. For example, the Hungarian algorithm computes a maximum weight matching of a bipartite graph for each target size~\cite{Kuhn10a}. %\greg{Should Algorithm 1 be mentioned in the proof?}
  		 %  (see Algorithm \ref{alg1}} \ali{any reference?}.
  	 % \haris{The Hungarian algorithm compute a maximum weight matching of a bipartite graph for each target size.}
  		\end{proof}

Similar approaches based on \textsc{Maximum Weight Matching} also lead to efficient algorithms for two variants of \textsc{Team Order}.
First, if the goal is to decide whether $T_1$ can beat $T_2$ with non-zero probability, the problem can be solved in polynomial time.
Specifically, for an instance represented as a graph $G$, we can consider the corresponding graph $G'$ in which edges with weight $0$ are removed. Then, $T_1$ can beat $T_2$ with non-zero probability if and only if $G'$ has a matching of size $\floor{\frac{n}{2}}+1$.
We state this result below.

\begin{corollary}
%For a team order problem instance represented as a graph $G = (T_1 \cup T_2, E, p)$, consider the corresponding graph $G'$ in which the weight-zero edges are removed. Then, $T_1$ can beat $T_2$ with non-zero probability if and only if $G'$ has a matching of size $\floor{n}{2}+1$.
Given the line-up of $T_2$, it can be decided in polynomial time whether there exists a line-up of $T_1$ that beats $T_2$ with a non-zero probability.
\end{corollary}					 				
Second, if the goal is to maximize the probability of winning \emph{all} the matches, the problem reduces to computing a weight maximizing matching, where the weights are the logarithm of the non-zero winning probabilities.
								
  									 \begin{proposition}\label{prop:allmatched}
  									%For $m=2$, 
									%\jiarui{I removed ``For $m=2$''. I think $m$ is not defined in the current version.}
									Given the line-up of $T_2$, the line-up of $T_1$ that maximizes the probability of winning \emph{all} the matches can be computed in polynomial time.
  										 \end{proposition}

		 	\section{Approximation Algorithm for Team Order}\label{sec:aprrox}
		
		\gregRN{As we have seen, in several cases finding a solution to \textsc{Team Order} is tractable. However, even though it resembles \textsc{Maximum Weight Matching}, its exact solutions are far more nuanced, which suggests its hardness.} In this section, we \gregRN{address the practical solvability of our problem by providing} a PTAS for \textsc{Team Order}.  Assuming  the input probabilities are bounded away from $0$ and $1$ by any arbitrary constant $\varepsilon > 0$, the PTAS computes a solution to \textsc{Team Order} whose winning probability is at most $\varepsilon$ less than that of  the optimal solution.
		
	\subsection{High-level Ideas}
		% 		\ali{Her I give a brief description about the algorithm and the idea, what do you think. }
		For any perfect matching $M=\{e_1,\ldots,e_n\}$ of $G=(T_1\cup T_2,E,p)$, let $X_M$ be a random variable counting the number of matches won by $T_1$. One may observe that $X_M$ follows a Poisson binomial distribution $PB(p_{e_1},\ldots,p_{e_n})$. 
		%Our proposed algorithm solves the following optimisation problem;
		Furthermore, \textsc{Team Order} can be written as the following optimization problem.
		\begin{equation*}
		\begin{aligned}
		%\max \quad &\Pr{X_M\ge \lfloor n/2\rfloor+1} = \sum_{\stackrel{S\subset\{1,..,n\}}{|S|\ge \lfloor n/2\rfloor+1}}\prod_{i\in S} p_{i,M(i)}\prod_{i\notin S}(1-p_{i,M(i)})\\
		\min_M \quad &\Pr{X_M\le \lfloor \frac{n}{2}\rfloor} \\%=1- \sum_{\stackrel{S\subset\{1,..,n\}}{|S|\ge \lfloor n/2\rfloor+1}}\prod_{i\in S} p_{i,M(i)}\prod_{i\notin S}(1-p_{i,M(i)}) \\
		\text{subject to:} \quad &  M ~ \text{is a perfect matching of } 
		G=(T_1\cup T_2, E, p) 
		\end{aligned}
		\end{equation*}
		% where   $p_{i,M(i)}$ is the winning probability player $i\in T_1$  against its pair $M(i)\in T_2$ in matching $M$.
		% jiarui: notation already defined.
		%  	 Alternatively, we can write the objective as:
		%  	   	\begin{equation*}
		%  	 \begin{aligned}
		%  	 \min \quad &\Pr{X_M\le \lfloor n/2\rfloor} =1- \sum_{\stackrel{S\subset\{1,..,n\}}{|S|\ge \lfloor n/2\rfloor+1}}\prod_{i\in S} p_{i,M(i)}\prod_{i\notin S}(1-p_{i,M(i)}) %\\
		%  	 %\text{subject to:} \quad &  M ~ \text{is a matching of } G=(T_1\cup T_2, E, p)
		%  	 \end{aligned}
		%  	 \end{equation*}  
		
		The main idea of our algorithm is as follows.
		First, we note that the number of matchings $M$ with $\Var{X_M}<\varepsilon^{-2}$ is bounded from above by a polynomial in $n$, when $\varepsilon$ is a constant. Hence, we can search over all such matchings to find out the optimal one among them.
		For the other matchings $M$ with a high variance $\Var{X_M} \ge \varepsilon^{-2}$, we use $\Phi\left(\frac{\lfloor \frac{n}{2} \rfloor-\Ex{X_M}}{\sqrt{\Var{X_M}}}\right)$ to approximate the objective function, where $\Phi(x)=(\frac{1}{\sqrt{2\pi}})\int_{-\infty}^x\mathrm{e}^{\frac{-y^2}{2}}dy$. Since $X_M$ is a Poisson binomial random variable, it holds that if $\Var{X_M}\ge \varepsilon^{-2}$, then
	\gregR{	\[
		\left | \Pr{X_M\le \lfloor \frac{n}{2} \rfloor} - \Phi\left(\frac{\lfloor \frac{n}{2} \rfloor-\Ex{X_M}}{\sqrt{\Var{X_M}}}\right) \right | \le \varepsilon.
		\]}
		Using the fact that $\Phi(x)$ is an increasing and continuous function in $x$, we get the following optimization problem as an approximation to the original one.
		\begin{equation*}
		\begin{aligned}
		\min_M \quad &\frac{\lfloor \frac{n}{2} \rfloor-\Ex{X_M}}{\sqrt{\Var{X_M}}} \\
		\text{subject to:} \quad &  M ~ \text{is a perfect matching of } 
		G=(T_1\cup T_2, E, p)
		\end{aligned}
		\end{equation*} 
		The objective function is still non-linear though, but it can be characterized by the mean and variance of $X_M$.
		Using the fact that for every matching $M$ we have $0\le \Var{X_M}\le \frac{n}{4}$ and $\Ex{X_M}\le n$, we can discretize the two dimensional space $\{(x, y): 0 \le x \le \frac{n}{4} \text{ and } 0 \le y \le n \}$ and design a search mechanism to eventually hit a matching that is close enough to the optimal matching.
		% (see Figure~\ref{fig:2}) \jiarui{I feel that Figure~\ref{fig:2} is not very informative... I suggest we remove it.}. 
		The search mechanism is based on an approximation algorithm solving a matching problem that involves both budget and rewards, which we will discuss next.

		\subsection{Preliminary Results}
		%An optimal line-up, denoted by $O$,  is a matching that minimises the losing probability or equivalently maximises  the winning probability  of $T_1$ against $T_2$. Recall that for any perfect matching $M=\{e_1,\ldots,e_n\}$ of $G$, $X_M$ is a random variable counting the number of matches won by $T_1$. One may observe that $X_M$ has Poisson binomial $PB(p_{e_1},\ldots,p_{e_n})$.  
		
		We introduce necessary preliminary results for designing the PTAS.
		It has two main ingredients. We apply a normal distribution estimation for a Poisson binomial  distribution, and an approximation algorithm for the following \textsc{Budgeted/Reward Matching} problem.
		{We assume that  every $p_e\notin\{0,1\}$ is bounded away from $0$ and $1$. 
		Define
		\gregR{$\delta = \min_{e\in E, p_e\notin \{0,1\}}\min\{p_e, 1-p_e\}. $}
		Then, we get that $\frac{1}{\delta}=\Theta(1)$. }	
  
		\paragraph{Approximation of Poisson Binomial Distribution.}
		We use a normal distribution estimation for a Poisson binomial  distribution to approximate $\Pr{X_M\le \lfloor \frac{n}{2} \rfloor}$, which is based on the following result. 	
  
		\begin{theorem}[\textnormal{\cite[Theorem 3.5]{poisson2019}}] \label{thm:approx2}
			Suppose that $X\sim \textit{PB}(p_1,\ldots,p_n)$ is a Poisson binomial random variable. Then, for every $1\le k\le n$,
			\gregR{\[\left|\Pr{X\le k}-\Phi\left(\frac{k-\Ex{X}}{\sqrt{\Var{X}}}\right)\right|\le \frac{1}{\sqrt{\Var{X}}}, \]
			 	where $\Phi(x)=(\frac{1}{\sqrt{2\pi}})\int_{-\infty}^x\mathrm{e}^{\frac{-y^2}{2}}dy$}.
		\end{theorem}
  
		An immediate application of Theorem \ref{thm:approx2} results in to the following corollary.
		\begin{corollary}\label{coor:1}
			Suppose that $X_M\sim PB(p_{e_1},\ldots,p_{e_n})$ is a Poisson binomial random variable corresponding to a matching $M=\{e_1,\ldots,e_n\}$ with     $\Var{X_M}\ge \varepsilon^{-2}$, for some $\varepsilon>0$. Then, 
			%  		\[
			%  	\Phi\left(\frac{\lfloor n/2\rfloor-\Ex{X_M}}{\sqrt{\Var{X_M}}}\right)-\varepsilon\le 	\Pr{X_M\le  \lfloor n/2\rfloor}\le \Phi\left(\frac{\lfloor n/2\rfloor-\Ex{X_M}}{\sqrt{\Var{X_M}}}\right)+\varepsilon,
			%  		\]
			\[
			\left | \Pr{X_M\le \lfloor \frac{n}{2} \rfloor} - \Phi\left(\frac{\lfloor \frac{n}{2} \rfloor-\Ex{X_M}}{\sqrt{\Var{X_M}}}\right) \right | \le \varepsilon,
			\]
			where $\Phi(x)=(\frac{1}{\sqrt{2\pi}})\int_{-\infty}^x\mathrm{e}^{\frac{-y^2}{2}}dy$.	
		\end{corollary}

		\paragraph{Budgeted Matching.}
		%The approximation  algorithm has two main ingredients. We apply  an approximation algorithm for the following \textsc{Budgeted/Reward Matching} problem, and a normal distribution estimation for a Poisson binomial  distribution.
		We will use approximation algorithms for the following \textsc{Budgeted Matching} problem as subroutines in our algorithm.
		
		\pbDef{\textsc{Budgeted Matching}}
		{A bipartite graph $G$, weight $w(e)$ and cost $c(e)$ for each edge $e$, and a budget $B$.}
		{Compute a perfect matching $M$ of $G$ that maximizes $w(M) :=\sum_{e\in M}w(e)$, subject to $c(M) :=\sum_{e\in M} c(e)\le B$.}
		Specifically, we are interested in the following weight and cost functions.
		For every $e\in E$, we let 
		$w(e)=p_e$  and $c(e)=p_e \cdot (1-p_e).$ 
		Hence, for every  matching $M$, we have 
		\begin{equation*}
		w(M) %=\sum_{e\in M}w(e)
		%=\sum_{e\in M}p_e
		=\Ex{X_M},
		\quad\text{and }\quad
		%  	\end{equation*}
		%  	and 
		%  	\begin{equation*}
		c(M) %=\sum_{e\in M}c(e)
		%=\sum_{e\in M}p_e \cdot (1-p_e)
		=\Var{X_M}.
		\end{equation*}
		We will henceforth stick to the above weight and cost functions, unless otherwise clarified.
		%
		%	\pbDef{\textsc{Reward Matching}}
		%	{A bipartite graph $G$, weight $w(e)$ and reward $c(e)$ for each edge $e$, and a threshold $R$.}
		%  	{Compute a perfect matching $M$ of $G$ that maximizes $w(M) :=\sum_{e\in M}w(e)$ while $c(M)=\sum_{e\in M} c(e)\ge R$.}	
		%
		%
		We use $I_b(G,w,c,B)$ to denote an instance of \textsc{Budgeted Matching}.
		%and $I_r(G,w,c,R)$ to denote instances of the above two problems, respectively.
		For convenience, we can also define a ``rewarded'' variant of \textsc{Budgeted Matching}, where we want the total cost to pass a threshold $R$, i.e., $c(M)=\sum_{e\in M} c(e)\ge R$, and we denote it by $I_r(G,w,c,B)$.
		{Since $0\le c(e)< 1$, we observe  that} 
		$I_r(G,w,c,B)$ is equivalent to $I_b(G, w, c', n-R)$, where $c'(e) = 1 - c(e)$ for every $e \in E$.
		Berger et al. \cite{Berger11}  designed a PTAS for the \textsc{Budgeted Matching} problem. 
		Using the same idea this PTAS is based on, we get the following. 
		
		\begin{lemma}\label{lem:budget}\label{lem:reward}
			Suppose that $G = (T_1 \cup T_2, E, p)$ is a \textsc{Team Order} instance, $w(e)=p_e$ and $c(e)=p_e \cdot (1-p_e)$ for each $e \in E$.
			Then, there is a polynomial-time algorithm to compute a feasible solution $M$ to $I_{b}(G,w,c,B)$ (respectively, $I_{r}(G,w,c,R)$) such that $w(M)\ge opt-2$, where $opt$ is the weight of optimal solution of $I_{b}(G,w,c,B)$ (respectively, $I_{r}(G,w,c,R)$). 
		\end{lemma}  

		\paragraph{Small / Large Variance Matchings.}
		We partition the set of edges into edges with fractional and binary weights; let $F=\{e\in E: p_e\notin\{0,1\}\} $ and $	\overline{F} =\{e\in E: p_e\in\{0,1\}\}.$
%		\[
%		F=\{e\in E: p_e\notin\{0,1\}\} 
%		\quad \text{ and } \quad 
%		\overline{F} =\{e\in E: p_e\in\{0,1\}\}.
%		\] 
		Fix an  arbitrary constant  $\varepsilon\in (0,1]$ and define
$\mathcal{M}^{+}(\varepsilon)= \left\{N\subset F:   \text{  $N$ is a minimum size matching with  $c(N)>\varepsilon^{-2}$} \right\}$,  
and
		\begin{align*}
		%&\mathcal{M}^{+}(\varepsilon)= \left\{N\subset F:   \text{  $N$ is a minimum size matching with  $c(N)>\varepsilon^{-2}$} \right\}.
	%	\\
		&\mathcal{M}^{-}(\varepsilon)= \left\{N\subset F:   \text{  $N$ is a  matching with  $c(N)\le \varepsilon^{-2}$} \right\}.
		\end{align*}

		Clearly, for every perfect matching $M$ on $G$, if $c(M)>\varepsilon^{-2}$, then there exists $N\in \mathcal{M}^+(\varepsilon)$ such that $M\cap N=N$. Similarly, if $c(M)\le \varepsilon^{-2}$, there exists $N\in \mathcal{M}^-(\varepsilon)$ such that $M\cap N=N$.
		
		For every matching $N \subset E$ and every subset of edges $E'\subseteq E$, let 
		%\ali{ $E'(\cdot)$ is  used for functions while $E'$ is a set. What about this $E'_N$?}
		%$E'(N)$ denote the set of all edges in $E'$  which do not share any endpoints with $N$.
		%
		$E'_N = \{ e \in E' : e \cap N = \emptyset \}$,
		i.e., $E'_N$ is the set of all edges in $E'$  that do not share \gregR{all} endpoints with $N$.
		We now define two families of bipartite graphs as follows. First,
		$\mathcal{G}^+(\varepsilon)=\{H=(T_1\cup T_2, N\cup E_N, p):
		N\in \mathcal{M}^{+}(\varepsilon) \}$. 
		Intuitively, we fix the matching $N$ and leave the unmatched part of the graph $G$ free.
		%
		%Then, let $\overline{F}=E-F$, which is the set of all the edges whose weights are either $0$ or $1$; 
%
		Then, we define $
		\mathcal{G}^-(\varepsilon)=\{H=(T_1\cup T_2, N\cup \overline{F}_N, p):
		N\in \mathcal{M}^{-}_\varepsilon \}.
		$
		This differs from $\mathcal{G}^+(\varepsilon)$, as we only consider 0/1-edges in the unmatched part of $G$.
		Note that for every perfect matching $M$ of $G$, if $c(M)>\varepsilon^{-2}$, there is $H\in \G^+(\varepsilon)$ such that $M\subset H$.
		Similarly, if $c(M)\le \varepsilon^{-2}$, then there is $H\in \G^-(\varepsilon)$ such that $M\subset H$.
		Next, we show that the size of these families of graphs is polynomially bounded.
		\begin{lemma}\label{obs:1}
			It holds that $|\mathcal{G}^+(\varepsilon)|\le n^{4\delta^{-1}\varepsilon^{-2}}$ and $|\mathcal{G}^-(\varepsilon)|\le n^{4\delta^{-1}\varepsilon^{-2}}$.
			%  	\[
			%  	|\mathcal{G}^+(\varepsilon)|\le n^{2(\varepsilon\delta_0)^{-2}},
			%  	\]
			%  	and 
			%  	\[
			%  	|\mathcal{G}^-(\varepsilon)|\le n^{2(\varepsilon\delta_0)^{-2}}.
			%  	\]
		\end{lemma}

		\subsection{The Algorithm}

			\begin{algorithm}[t]
		\caption{{\sc $\varepsilon$-Approximation Algorithm} }\label{alg2}
			%\begin{algorithmic}[1]
			
			\BlankLine
			\KwIn{a \textsc{Team Order} instance $G =(T_1\cup T_2, E, p)$.}
			\KwOut{an $\varepsilon$-approximate solution to \textsc{Team Order}.}
			
			\BlankLine
			
			$\varepsilon \leftarrow \frac{\varepsilon}{4}$;
			
			\For{ $H \in \mathcal{G}^-(\varepsilon)$ \label{ln:for-1-start}}{
				
				%{Compute a maximum weight perfect matching of $H$ (if exists), say $M_H$;}
				$M^* \leftarrow$ a maximum weight perfect matching of $H$; 
				%\jiarui{I added ``weight'' here (and in proof of Thm~\ref{thm:alg2}). I suppose that's what you mean? }\ali{ yes, you are right.   by definition, matching of  $H$ has $N\in \M^{-}_\varepsilon$ plus some edges from $\bar{F}(N)$	whose weighrts are zero or one.}
				
				%If $M_H$ has higher winning probability than the previous one, keep it as $M^{(1)}_\varepsilon$;		
				
				\If{$M^*$ exists and has a higher winning probability than $M^{(1)}_\varepsilon$ (or $M^{(1)}_\varepsilon = null$)}
				{
					$M^{(1)}_\varepsilon \leftarrow M^*$; \label{ln:for-1-end}
				} 
				
			} 
			
			%{$val \leftarrow +\infty$;}
			
			{$x_i \leftarrow \varepsilon^{-2}+ \frac{i}{n}$ for each $i=0,\ldots, \frac{n^2}{4}$;}
			
			\For{$H \in \mathcal{G}^+(\varepsilon)$ \label{ln:for-2-start}}{
				
				\For{ $i=1,\ldots,  \frac{n^2}{4}$}{
					$M_i^* \leftarrow $ a solution to $I_b(H,w,c,x_i)$ such that $w(M_i^*) > opt - 2$; \label{ln:Mi-st-1}
					\tcp{Lemma \ref{lem:budget}}
					
					\If{ $M_i^*$ exists and $w(M_i)<\lfloor \frac{n}{2} \rfloor$}{
						{ $M_i^* \leftarrow $ a solution to $I_r(H,w,c,x_{i-1})$ such that $w(M_i^*) > opt - 2$;} \label{ln:Mi-st-2}
						\tcp{Lemma \ref{lem:reward}}
					}
					
					%  		\If{$M_i$ exists and $val>\frac{\lfloor n/2\rfloor -w(M_i)}{\sqrt{c(M_i)}}$} {
					
					\If{$M_i^*$ exists and $\frac{\lfloor \frac{n}{2}\rfloor -w(M^{(2)}_\varepsilon)}{\sqrt{c(M^{(2)}_\varepsilon)}} > \frac{\lfloor \frac{n}{2} \rfloor -w(M_i^*)}{\sqrt{c(M_i^*)}}$ (or $M^{(2)}_\varepsilon = null$)} {
						
						%$val \leftarrow \frac{\lfloor n/2\rfloor -w({{M_{i}})}}{\sqrt{c(M_{i})}}$;
						
						{$M^{(2)}_\varepsilon \leftarrow M_i^*$;} %\Comment{lines 14-17 searches for a matchings with min. val }
						\label{ln:for-2-end}
					}
				}
			}
			\Return{ $M^{(1)}_\varepsilon$ or $M^{(2)}_\varepsilon$ whichever has the higher winning probability.} 
			%\end{algorithmic}
		\end{algorithm}
		
		Now we discuss our approximation algorithm, Algorithm~\ref{alg2}.
		Theorem~\ref{thm:alg2} shows that the algorithm produces an $\varepsilon$-approximate solution to \textsc{Team Order} in polynomial time.
\harisnew{The proof relies on  Lemma \ref{lem:budget} and Lemma~\ref{obs:1}. }

		\begin{theorem}
			\label{thm:alg2}
			Algorithm~\ref{alg2} computes an $\varepsilon$-approximate solution to \textsc{Team Order} and runs in time $n ^{O(\delta^{-1}\varepsilon^{-2})}$, where $\delta=\min_{ p_e\notin\{0,1\}}\min_{e\in E}\{p_e, 1-p_e\}$.  
			%\jiarui{I changed $3 \varepsilon$ to $\varepsilon$, and also changed Algorithm~\ref{alg2} accordingly.}
		\end{theorem}

  {\section{Winning Probability of a Maximum Weight Matching}}
  	%\ali{I improved the write-up, removed some explanations, and added the preliminary results to an appendix}.
	In this section we investigate the winning probability of a maximum weight matching.
  	Our result provides a lower bound for the winning probability of any maximum weight matching compared with that of the optimal line-up. In particular, the result shows that  \gregR{a} sufficiently large/small  maximum weight matching \gregR{performs} almost as well as the optimal line-up. 
  In what follows, we view $G=(T_1\cup T_2, E, p)$ as a weighted bipartite graph where for every $e\in E$, $p_e$ is the weight of $e$.  For every matching $M$, we use $w(M)$ to denote its weight. In addition, we assume that the size of $G$ is sufficiently large.
	
  	\begin{theorem}\label{thm:main}
	  	
  	 	Let $M^*$ be a maximum weight matching and let $O$ be an optimal line-up. Then,
  	 	\begin{itemize}
  	 			\item[(1)] If $w(M^*)= \frac{n}{2}\pm f(n)\sqrt{n}$, where $f(n)\in [1, \frac{\sqrt{n}}{2}]$ is any non-decreasing  function in $n$, then 
  	 		$\Pr{T_1 \text{wins under $O$} }\le \Pr{T_1 \text{wins under $M^{*}$} }+ \mathrm{e}^{-2f^2(n)}$.
  	 		
  	 		\item[(2)] If $w(M^*) \in [\frac{n}{2}-\sqrt{n\log n},\ \frac{n}{2}+\sqrt{n\log n}]$, then
  	 		\begin{align*}
  	 	&	\Pr{T_1 \text{wins under $O$} } \\
\le&\Pr{T_1 \text{wins under $M^{*}$} }+ \frac{(4+o(1))}{n+1}\sum_{e\in M^*}(p_e-\frac{1}{2})^2.
  	 		\end{align*}
  	 	\end{itemize}	 	
  	 \end{theorem} 

In the proof of Theorem \ref{thm:main} we  rely on the following results.% from \cite{Dubhashi}, \cite[Theorem 2.1]{poisson2019}, and \cite[Theorem 1]{EHM'91}.

%\gregR{In the proof of Theorem \ref{thm:main} we will rely on the following facts.}
%
%
\begin{theorem}[\cite{Dubhashi}]\label{hoeffding}
	\gregR{Suppose that $n$ is a given positive integer and  let $X\sim PB(p_1,\ldots, p_n)$  be a Poisson binomial  random variable. Then, we have that  
	$
	\Pr{X\ge \Ex{X}+\delta}\le e^{\frac{-2\delta^2}{n}},
	$
	and 
	$
	\Pr{X\le \Ex{X}-\delta}\le e^{\frac{-2\delta^2}{n}}.
	$ }
\end{theorem}

	 \begin{theorem}[\textnormal{\cite[Theorem 2.1]{poisson2019}}]\label{stoch}
		\gregR{ 	Let $X\sim PB(p_1,\ldots,p_n)$ and let $\bar{p}=\sum_{i=1}^n \frac{p_i}{n}$. Define 
			$Y\sim Bin(n, \bar{p})$. Then, (1) for every $0\le k\le n\bar{p}-1$,
				$
				\Pr{X\le k}\le \Pr{Y\le k}$, and
				(2) for every $n\bar{p}\le k\le n$, 	$
				\Pr{X\le k}\ge \Pr{Y\le k}$.
			}
		\end{theorem}
	
		\begin{theorem}[\textnormal{\cite[Theorem 1]{EHM'91}}]\label{thm:approx}
		\gregR{ 	Suppose that  $X\sim PB(p_1,\ldots,p_n)$, and $\bar{p}=\sum_{i=1}^np_i/n$. Also,
			let $Y\sim Bin(n, \bar{p})$ is a binomial probability distribution. Then,			\[
			\max_{A\subseteq \{0,\ldots, n\}}	\left|\Pr{X\in A}-\Pr{Y \in A}\right|\le \frac{1-\bar{p}^n-(1-\bar{p})^n}{(n+1)\bar{p}(1-\bar{p})}\sum_{i=1}^n(p_i-\bar{p})^2.
			\]}

		\end{theorem}

	    	 We prove the first and the second parts of the theorem separately next.  
	 
\subsection*{Part 1. When $w(M^*)=\frac{n}{2}\pm f(n)\sqrt{n}$}\label{sub:1}	
\begin{proof}
		Let $M= \{e_1,\ldots,e_n\}$ be an arbitrary matching and $X_{M}$ be a random variable that counts the number of games won by $T_1$ under line-up $M$. Then $X_{M}$ follows Poisson binomial distribution $PB(p_{e_1},\ldots,p_{e_n})$. 
		, where $M= \{e_1,\ldots,e_n\}$. 
	 Thus,  $\Ex{X_{M}}=w(M)$.
	 Let us first assume that $w(M^*)=\frac{n}{2}-f(n)\sqrt{n}$. Then, for every matching $M$, including the optimal line-up $O$, we have  $w(M)\le w(M^*)$. Moreover, we have $f(n)\sqrt{n}=\frac{n}{2}-w(M^*)\le \frac{n}{2}-w(M)$, and  
	\begin{align*}
	\Pr{T_1 \text{ wins under } M}&= \Pr{X_{M}\ge \lfloor \frac{n}{2} \rfloor+1}\nonumber\\
	&\le \Pr{X_{M}\ge w(M)+(\frac{n}{2} -w(M))}\nonumber\\
	&=\Pr{X_{M}\ge \Ex{X_{M}}+f(n)\sqrt{n}}\le \mathrm{e}^{-2(f(n)\sqrt{n})^\frac{2}{n}}=\mathrm{e}^{-2f(n)^2},
	\end{align*}
using a concentration bound for Poisson binomial random variables (e.g., see Theorem \ref{hoeffding}).
Following that upper bound, if  $w(M^*)=\frac{n}{2}-f(n)\sqrt{n}$, then
	\begin{align}\label{case:1}
\Pr{\text{$T_1$ wins under  $O$ }}&\le \mathrm{e}^{-2f(n)^2}\le
\Pr{\text{$T_1$ wins under  $M^{*}$ }}+\mathrm{e}^{-2f(n)^2}.
\end{align}
Next, we consider the case where $w(M^*)=\frac{n}{2}+f(n)\sqrt{n}$. 
Define random variable $Y_{M^*}$ that counts the number of games lost under $M^*$. Then $Y_{M^*}$ follows Poisson binomial distribution $PB(1-p_{e_1},\ldots,1-p_{e_n})$, where we let $M^*=\{e_1,\ldots,e_n\}$.
One can check that $\Ex{Y_{M^*}}=n-w(M^*)=\frac{n}{2}-f(n)\sqrt{n}$. 
\begin{align*}
\Pr{\text{$T_1$ loses under  $M^*$ }}&= \Pr{Y_{M^*}\ge \lfloor \frac{n}{2}\rfloor+1}\\
&\le \Pr{Y_{M^*}\ge \Ex{Y_{M^*}}+(\frac{n}{2}-\Ex{Y_{M^*}})}\le \mathrm{e}^{-2f(n)^2},
\end{align*} 
where we have applied the same concentration bound as the previous case.
Hence, 
\begin{align*}
\Pr{\text{$T_1$ wins under  $M^*$ }}=1-\Pr{\text{$T_1$ loses under $M^*$ }}\ge 1-\mathrm{e}^{-2f(n)^2}.
\end{align*}
Thus,
\begin{align}\label{case:2}
\Pr{\text{$T_1$ wins under  $O$ }}\le  1\le 
\Pr{\text{$T_1$ wins under  $M^*$ }}+\mathrm{e}^{-2f(n)^2}
\end{align}
Hence, combining \eqref{case:1} and \eqref{case:2} gives the first part of Theorem \ref{thm:main}. 

	\end{proof}

  	 \subsection*{Part 2. When $w(M^*) \in [\frac{n}{2}-\sqrt{n\log n},\ \frac{n}{2}+\sqrt{n\log n}]$}%\label{sub:2}

  	 \begin{proof}
 	 	Let us first consider the case  where $w(M^*)\in [\frac{n}{2}-\sqrt{n\log n}, \frac{n}{2}-1)$.
  	 	Define random variables $X_O$ and $X_{M^*}$ that count the number games won by $T_1$ under $O$ and $M^*$, respectively. Moreover, define  binomial random variables
  	 	$Z_O\sim Bin(n, \frac{w(O)}{n})$ and $Z_{M^*} \sim Bin(n, \frac{w(M^*)}{n})$. Notice that 
  	 	 $w(O)\le w(M^*)$ and hence $Z_{M^*}$ stochastically dominates $Z_{O}$ (i.e, $ \Pr{Z_O\le \frac{n}{2}} \ge \Pr{Z_{M^*}\le \frac{n}{2}}$). 
  	 	 Since $w(M^*)<\frac{n}{2}$, we apply  the stochastic dominance between the Poisson and binomial random variables (e.g., see \gregR{Theorem \ref{stoch} (2)}  ) and we have that   
  	 	\begin{align*}
  	 		\Pr{\text{$T_1$ loses under $O$}}=\Pr{X_O\le \frac{n}{2}}\ge  \Pr{Z_O\le \frac{n}{2}} \ge \Pr{Z_{M^*}\le \frac{n}{2}},
  	 	\end{align*} 
  	 	On the other hand,  the optimal line-up $O$ minimizes the losing probability of $T_1$ and
  	 	hence, by above inequality we have that 
		   \begin{align*}
			&\Pr{\text{$T_1$ loses under $M^*$}} \\
			&\qquad\qquad= \Pr{X_{M^*}\le \frac{n}{2}}\ge \Pr{\text{$T_1$ loses under $O$}}\ge \Pr{Z_{M^*}\le \frac{n}{2}}.
		   \end{align*}
  	%	\[
  	% \Pr{\text{$T_1$ loses under $M^*$}}= 
	% \Pr{X_{M^*}\le \frac{n}{2}}\ge \Pr{\text{$T_1$ loses under $O$}}\ge \Pr{Z_{M^*}\le \frac{n}{2}}.
  	% 	\]  
  	 	Applying the above inequality and  Theorem \ref{thm:approx} results in
  	 	\begin{align*}
  	 	&\Pr{\text{$T_1$ wins under $O$}}-\Pr{\text{$T_1$ wins under $M^*$}}\\
  	 	&\qquad\qquad= (1-\Pr{\text{$T_1$ wins under $M^*$}}) - (1- \Pr{\text{$T_1$ wins under $O$}})\\
  	 	&\qquad\qquad=\Pr{\text{$T_1$ loses under $M^*$}}-\Pr{\text{$T_1$ loses under $O$}}\\
  	 	&\qquad\qquad\le \Pr{X_{M^*}\le \frac{n}{2}}-\Pr{Z_{M^*}\le \frac{n}{2}} \\
		&\qquad\qquad\le \frac{1-(\bar{p})^{n}-(1-\bar{p})^n}{(n+1)(1-\bar{p})\bar{p}}\sum_{e\in M^*} (p_e-\bar{p})^2,
  	 	\end{align*}
  	  where $\bar{p}=\frac{w(M^*)}{n}$.
  	 	Since we have $n\bar{p}\in (\frac{n}{2} -\sqrt{n\log n}, \bf{\frac{n}{2}})$, and $n$ is an asymptotically large,  we have $\bar{p}\approx \frac{1}{2}$ and thus 
  	 	\begin{align*}\label{upper}
  	 	\frac{1-(\bar{p})^{n}-(1-\bar{p})^n}{(n+1)(1-\bar{p})\bar{p}}\sum_{e\in M^*}(p_e-\bar{p})^2\le  \frac{(4+o(1))}{n+1}\sum_{e\in M^*} (p_e-\frac{1}{2})^2.
  	 	\end{align*}
  	 	Therefore, if  $w(M^*)\in [\frac{n}{2}-\sqrt{n\log n}, \frac{n}{2}-1)$, then
  	 	\[
  	 \Pr{\text{$T_1$ wins under $O$}}\le \Pr{\text{$T_1$ wins under $M^*$}} + \frac{(4+o(1))}{n+1}\sum_{e\in M^*}(p_e-\frac{1}{2})^2.
  	 	\]

%%\greg{It might be useful to include the statement of Theorems A2 and A3 in the main body of the paper (to allow a reader to understand the proof without reading the appendix)}
%     
  	 To derive the same upper bound for the case where $w(M^*)\in [\frac{n}{2}, \frac{n}{2}+\sqrt{n\log n}]$, we define random variables that count the number of  games lost by $T_1$ and the same technique for the above case follows.
\end{proof}

			\section{Conclusion}
			
We proposed the {\sc Team Order} problem, which naturally captures several strategic scenarios in information systems and team competitions. 
% We have demonstrated that it is NP-complete in the general case. 
\gregR{We have shown that in the case in which the input probabilities are limited to three values (including 0) it is tractable and have shown that it is possible to efficiently compute a line-up which is close to the optimal in terms of the probability of winning, which is useful when the information about the players' relative strength is limited (e.g., if it is only known when a player is ``strong'' or ``weak'' against an opponent).}
 One of our central results is a PTAS for the \textsc{Team Order} problem. We note that while we focused on the probability of winning against more than a half of opposing players, our results hold for any such threshold.

\gregR{We conclude by highlighting some important directions for future work.} 
First, the complexity of solving {\sc Team Order} exactly is open.  We believe that this is a challenging question that also has implications on the related problem of {\sc Colored Bipartite Matching}. It is known that it  is NP-complete when $d$ is a variable~\cite{GeSz11a}.  However, the complexity of this problem is open if $d$ is a constant larger than $2$, or if $d=2$ but the graph is incomplete (which corresponds to $\{\alpha, \beta, 0\}$)~\cite{YMS02a}. This motivates further study between the connections of the two discussed problems. %\gregRN{Further, we note that while it is not known if finding an exact solution to {\sc Team Order} is possible in polynomial time, we demonstrate that finding a solution suited for practical purposes is tractable.}
    It is also not known whether {\sc Team Order} admits a fully polynomial-time approximation scheme (FPTAS). Resolving this question would be a strong improvement over our results.
%
   % \gregRN{We note that we assumed the full knowledge of the probability of winning. Relaxing this assumption would be an interesting follow-up.}

   While our result show the complexity of computing a best response to the opponents line-up, it is natural to study the extension in which multiple teams strategize. %In the Appendix we provide preliminary results on this topic, e.g., regarding finding Nash equilibria in the setting with two teams.}
   Regarding sport events, it would also be interesting to see if the results change under other natural assumptions, such as  all of the players having an objective level of skill. 
   %E.g., in chess contestants have corresponding Elo ratings. Then, we might assume that the pairwise winning probabilities is monotonic regarding this measurement.% One can also assume even more structured winning probabilities. 
   For example, if a player $i$ has better skill than a player $j$, then $i$ might always have a better probability of winning against any player $k$ than $j$'s probability of beating $k$.

\section*{Acknowledgments}

This work was supported by the NSF-CSIRO grant on “Fair Sequential Collective Decision-Making" (Grant  No. RG230833) and by DSTG under the project ``Distributed multi-agent coordination for mobile node placement.'' (Grant  No. RG233005). 
This project has also received funding from the European 
    Research Council (ERC) under the European Union’s Horizon 2020 
    research and innovation programme (grant agreement No. 101002854). 
Grzegorz Lisowski acknowledges support by the European Union under the Horizon Europe project %\href{https://perycles-project.eu/}
{Perycles}  (Participatory Democracy that Scales). 

    \vspace{2em}

	\includegraphics[width=0.4\columnwidth]{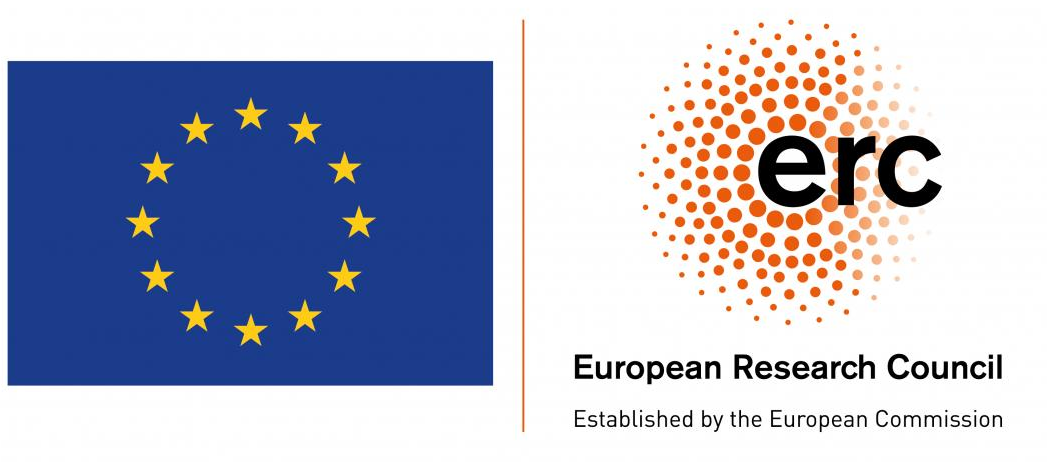}
~
\hfill
\raisebox{.8cm}{\includegraphics[width=0.4\columnwidth]{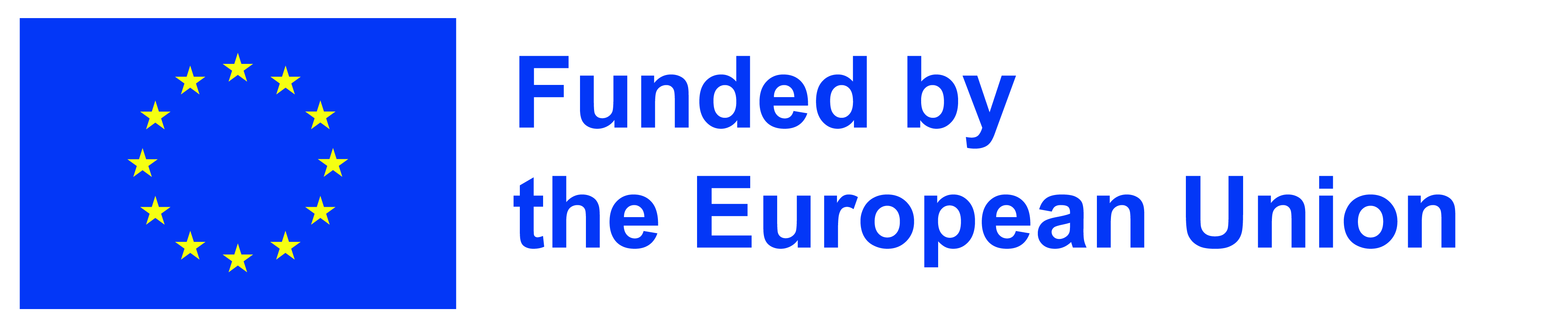}}

\bibliographystyle{ACMReferenceFormat} 
 \bibliography{abb,group,haris2B,aziz,grzegorz}

 %\bibliography{../abb,../group,../aziz}%,../haris2,../aziz,../grzegorz}

 %\bibliography{abb,group,haris2,aziz,grzegorz}
\

\newpage

\section*{Appendix}

%	 \begin{proposition}\label{prop:allmatched}
%  									%For $m=2$, 
%									%\jiarui{I removed ``For $m=2$''. I think $m$ is not defined in the current version.}
%									Given the line-up of $T_2$, the line-up of $T_1$ that maximizes the probability of winning \emph{all} the matches can be computed in polynomial time.
%  										 \end{proposition}
\section*{Proof of Lemma~\ref{lem:budget}}

\begin{proof}
	According to \cite[Lemma 4]{Berger11}, one can find in polynomial time a feasible solution $M$ to a \textsc{Budgeted Matching} instance $I_b(G,w,c,B)$, such that  $w(M)=\sum_{e\in M}w(e)\ge opt - 2w_{max}$, where $w_{max}$ is the highest edge weight.
	Now that $w(e)=p_e$ for each $e \in E$, it follows immediately that $w(M)\ge opt-2$.
\end{proof}

\section*{Proof of Proposition~\ref{pro:winning}}

  	\begin{proof}
  The algorithm is via dynamic programming. Let $p_j$ be the probability of the first team winning the \gregR{$j^{\text{th}}$} match according to the line-ups. 
  We keep track of an array $A$, where $A[i,j]$ denotes the probability of $i$ wins after the first $j$ matches. 
  In such an array, $A[0,1]=1-p_1$ and \harisnew{$A[0,0]=1$, $A[i,0]=$0 for $i \neq 0$}, 
\harisnew{and $A[0,j] = \prod_{k=1}^{j} (1-p_k)$ for all $j \in [n].$}
 Furthermore, we can compute a given entry  $A[i,j]$  by $A[i,j]=A[i-1,j-1] \cdot p_j+A[i,j-1] \cdot (1-p_j)$. 
  The array takes size $n(n+1)$ so the probability of first team winning a given number of matches $k$ can be computed by reading off the entry $A[k,n]$. In order to compute the probability of the first team winning more matches, we simply need to find the sum $\sum_{k\geq \floor{\frac{n}{2}}+1}A[k,n].$

   \end{proof}

\section*{Proof of Lemma~\ref{obs:1}}

\begin{proof}
By definition of  $ \mathcal{M}^-(\varepsilon)$,   for every $N \in \mathcal{M}^-(\varepsilon)$, we have  $N \subset F$  and hence $c(e) =  p_e \cdot (1-p_e) > 0$ for all $e \in N$ according to the definition of $F$.
	By assumption, we have that 
	\[
	c(e)=p_e \cdot (1-p_e) \ge \delta(1-\delta)\ge \frac{\delta}{2},
	\]
	which follows from definition of $\delta$  and the fact that $ 1-\delta\ge \frac{1}{2}$.
	Hence, $|N| \le \frac{2c(N)}{\delta} \le 2\cdot \delta^{-1}\cdot\varepsilon^{-2}$. It follows that
	
	%\[
	\begin{equation*}
	\begin{split}
	|\mathcal{M}^-(\varepsilon)|\le \sum_{s=1}^{\lceil 2\delta\varepsilon^{-2}\rceil}{\binom{|F|}{s}}\le 
	\lceil2\delta^{-1}\varepsilon^{-2}\rceil\cdot {\binom{|F|} {2} \delta\varepsilon^{-2}}\le \\
	{n^{4\delta^{-1}\varepsilon^{-2}}},
	\end{split}
	\end{equation*}
%	\] 
	where we use the fact that $|F|\le n^2$.
	According to the definition of $\G^-(\varepsilon)$, we have $|\G^-(\varepsilon)|\le |\mathcal{M}^-(\varepsilon)|$; hence, $|\mathcal{G}^-(\varepsilon)|\le n^{4\delta^{-1}\varepsilon^{-2}}$.  
	
	Similarly, consider every $N \in \mathcal{M}^+(\varepsilon)$. 
	By definition, every $N \in \mathcal{M}^{+}(\varepsilon) $ is a minimum size matching with $c(N)>\varepsilon^{-2}$.
	This also implies that $|N| \le 2\delta^{-1}\varepsilon^{-2}$. A similar argument results in  
	$|\mathcal{G}^+(\varepsilon)|\le n^{4\delta^{-1}\varepsilon^{-2}}$.	
	
\end{proof}

\section*{Proof of Proposition~\ref{prop:allmatched}}
             
  										 \begin{proof}
  										 Construct a bipartite graph $G' = (T_1\cup T_2, E, w)$, where $T_1$ and $T_2$ are the two independent sets of the graph and $\{t_1^i,t_2^j\}\in E$ if $t_1^i$ has non-zero probability of beating $t_2^j$. 
										 The weight of each edge $e = \{t_1^i,t_2^j\}$ is $w_e = \log(p(t_1^i,t_2^j)$.
										 %the log of the probability of $t_1^i$ beating $t_2^j$. 
										 Consider a maximum weight perfect matching in $G'$. Then the total weight $\sum_{\{t_1^i,t_2^j\}\in M}\log(p(t_1^i,t_2^j))$ is maximized among all perfect matchings. 
										 %Note that any matching that maximizes $\sum_{\{t_1^i,t_2^j\}\in M}\log(p(t_1^i,t_2^j))$ corresponds to an assignment $M$ if players of $T_1$ to players of $T_2$ that maximizes the term 
										 Equivalently, this means that the quantity 
  										 $\prod_{\{t_1^i,t_2^j\}\in M}(p(t_1^i,t_2^j))$ is also maximized, which is exactly to the probability of $T_1$ winning the competition. 
  									Hence, the problem reduces to computing a maximum weight perfect matching, which can be solved in $O(n^3)$ via the Hungarian algorithm~\cite{kuhn55hungarian}.
  											 \end{proof}

\section*{Proof of Theorem~\ref{thm:alg2}}
%\begin{theorem}
%			\label{thm:alg2}
%			Algorithm~\ref{alg2} computes an $\varepsilon$-approximate solution to \textsc{Team Order} and runs in time $n ^{O(\delta^{-1}\varepsilon^{-2})}$, where $\delta=\min_{ p_e\notin\{0,1\}}\min_{e\in E}\{p_e, 1-p_e\}$.  
%			%\jiarui{I changed $3 \varepsilon$ to $\varepsilon$, and also changed Algorithm~\ref{alg2} accordingly.}
%		\end{theorem}
		\begin{proof}
			Assume that $O$ is the optimal solution to the \textsc{Team Order} instance. If $c(O)\le \varepsilon^{-2}$, then there exists $N^* \in \mathcal{M}^-(\varepsilon)$ such that $N^* \cap O = N^*$. 
			Moreover, since $N^*$ contains all the edges with weight $0$ or $1$ in $O$, it is not hard to see that $O$ must also be a maximum weight matching of $H^* =(T_1\cup T_2, N^* \cup \overline{F}_{N^*}, p)$.
			Hence, the for-loop at Lines~\ref{ln:for-1-start}--\ref{ln:for-1-end} is able to find out $O$.  
			
			%	% and hence there is $H\in \G^{-}(\varepsilon)$ such that $O\subset H=(T_1\cup T_2, N\cup \overline{F}(N), p)$.
			%	%, where $\overline{F}$ is the set of edges with $0$ and $1$-weight edges. 
			%  	Also, we deduce that  $O$ dose not have any uncertain edges except including $N$. 
			%	Therefore, a maximum weight perfect matching for $H$ has the same number of $1$ and $0$-weight edges as $O$. 
			%	On the other hand, by Theorem \ref{thm:main} \jiarui{do you mean Proposition~\ref{pro:winning}} we are able to compute the   winning probability corresponding to a  given perfect matching (lienup). Since $O$ maximises the winning probability, the algorithm finally outputs $M^{(1)}_\varepsilon$ as an optimal line-up.
			
			Let us then assume $c(O)> \varepsilon^2$ in the remainder of the proof.
			In this case there exists $N^* \in \mathcal{M}^+(\varepsilon)$ such that $O\cap N^*=N^*$ and hence there is $H^*=(T_1\cup T_2, N^*\cup E_{N^*})\in \G^{+}(\varepsilon)$ so that $O\subset H^*$. Therefore,  every perfect matching  $M_{H^*}$ of $H^*$ contains $N^*$ and we have that $c(M_{H^*})\ge c(N^*)>\varepsilon^{-2}$.  On the other hand, by  the definition of $x_i$'s in the algorithm, there exists index $j$ such that 
			\gregR{$c(O)\in[x_{j-1}, x_j]\subset(\varepsilon^{-2}, \frac{n}{4}]$},
			which follows from the fact that $c(O)=\Var{X_{O}}\le \frac{n}{4}$ as  $X_{O}$ is a Poisson binomial random variable. Also note that $x_{j}-x_{j-1}= \frac{1}{n}$.
			%Moreover, if there is another $M_H$ with $c(M_H)\in [x_{j-1}, x_j]$, then $c(M_H)=c(O)$. Otherwise
			%By the numerical width assumption and the definition of $x_i$'s in the algorithm we have that  for every matching there is a unique interval so that $c(M)\in [x_{i-1}, x_{i}]$. Moreover, if there is another matching $M'$  with $c(M')\in [x_{i}, x_{i+1}]$, then $c(M)=c(M')$. Because we have that 
			%	$|c(M_H)-c(O)|\ge 2^{-l_0}> x_{j}-x_{j-1}=2^{-l_0-1}$. %Therefore, there are 
			%$H \in \mathcal{G}$ and index $j$ such that  $O\subset H$ and  
			%	\[
			%\varepsilon^{-2}\le x_{j-1}\le c(O)\le x_j.
			%	\] 
			Now, let us execute Line~\ref{ln:Mi-st-1} with $H = H^*$.
			Since $O$ is a feasible solution for $I_b(H^*,w,c,x_j)$, $M^*_j$ exists and can be computed in polynomial time according to Lemma \ref{lem:budget}. 
			We consider the following two cases.
			
			\paragraph{Case 1.} $w(M^*_j)\ge \lfloor \frac{n}{2} \rfloor$.
			%\begin{itemize}
			%\item 
			We have 
			\begin{align}\label{lab:8}
			w(M^*_j)\ge w(A)-2\ge w(O)-2,
			\end{align}
			where $A$ is an optimal solution to $I_b(H^*,w,c,x_j)$.
			Moreover, $c(M^*_j)\le x_j$ as $M_j^*$ is a feasible solution to $I_b(H^*,w,c,x_j)$, so it must be that
  
			\begin{align}\label{lab:8'}
			\varepsilon^{-2}\le c(M^*_j)\le x_j\le  c(O)+x_{j}-x_{j-1}= c(O)+ \frac{1}{n},
			\end{align}
			which follows from $c(O)\in [x_{j-1}, x_j]$ and $x_{j}-x_{j-1}=\frac{1}{n}$.
			Using  
			Inequalities (\ref{lab:8}) and (\ref{lab:8'}) we will have that 
			\begin{align}
			\frac{\lfloor \frac{n}{2}\rfloor-w(M^*_j)}{\sqrt{c(M^*_j)}}
			&\le\frac{\lfloor \frac{n}{2}\rfloor-w(M^*_j)}{\sqrt{c(O)+\frac{1}{n}}} & \text{(as $w(M^*_j) \ge \lfloor \frac{n}{2}\rfloor$)} \nonumber\\
			&\le \left(1-\frac{1}{2n\cdot c(O)}\right) \frac{\lfloor \frac{n}{2}\rfloor-w(M^*_j)}{\sqrt{c(O)}} & \text{(as $(1+x)^{-\frac{1}{2}}\ge (1-\frac{x}{2})$)}\nonumber\\
			&\le  \frac{\lfloor \frac{n}{2} \rfloor-w(M^*_j)}{\sqrt{c(O)}} + \frac{1}{2 n \cdot c(O)}\cdot \frac{w(M^*_j)-\lfloor \frac{n}{2} \rfloor}{\sqrt{c(O)}}\nonumber\\ 
			&\le \frac{\lfloor \frac{n}{2} \rfloor-w(O)+2}{\sqrt{c(O)}} + \frac{(n\cdot \varepsilon)}{n}  &\hspace{-10mm} \text{(as $w(M^*_j)\le n$ and $c(O) > \varepsilon^{-2}>1$)} \nonumber\\
			& \leq \frac{\lfloor \frac{n}{2}\rfloor-w(O)}{\sqrt{c(O)}}+\frac{2}{\sqrt{c(O)}} +\varepsilon\nonumber\\
			&\le   \frac{\lfloor \frac{n}{2}\rfloor-w(O)}{\sqrt{c(O)}} + 3\varepsilon 
			\label{lab:5}
			%   \frac{\lfloor n/2\rfloor-w(O)}{\sqrt{c(O)+\Delta}}+\frac{2}{\sqrt{c(O)+\Delta}}
			%  \nonumber\\
			%&\le  
			%\frac{1}{\sqrt{1+\Delta/c(O)}}\cdot  \frac{\lfloor n/2\rfloor-w(O)}{\sqrt{c(O)}}+\frac{2}{\sqrt{c(O)}}\nonumber\\
			% &\le \frac{\lfloor n/2\rfloor-w(O)}{\sqrt{c(O)}}+\frac{\Delta}{2c(O)}\cdot \frac{w(O)-\lfloor n/2\rfloor}{\sqrt{c(O)}}+2\varepsilon. & \text{(as $c(O) > \varepsilon^{-2}>1 $)}\nonumber\\
			%  & \le \frac{\lfloor n/2\rfloor-w(O)}{\sqrt{c(O)}}+\frac{1}{2n}\cdot n \cdot \varepsilon+2\varepsilon & 
			%  & \le \frac{\lfloor n/2\rfloor-w(O)}{\sqrt{c(O)}}+ 3\varepsilon
			%  \label{lab:5}
			\end{align}
			Furthermore, since $w(M^*_j)\ge \lfloor \frac{n}{2} \rfloor$, Line~\ref{ln:Mi-st-2} will not be executed, and subsequently, we obtain a perfect matching $M^{(2)}_\varepsilon$ with the minimum corresponding value. According to the above upper bound, we have  
			\begin{align}\label{lab:4}
			\frac{\lfloor \frac{n}{2} \rfloor-w(M^{(2)}_\varepsilon)}{\sqrt{c(M^{(2)}_\varepsilon)}}\le \frac{\lfloor \frac{n}{2} \rfloor-w(M^*_j)}{\sqrt{c(M^*_j)}}\le 
			\frac{\lfloor \frac{n}{2} \rfloor-w(O)}{\sqrt{c(O)}}+3\varepsilon.
			\end{align}
			Recall that $w(O)=\Ex{X_O}$ and $\varepsilon^{-2}<c(O) = \Var{X_O}$. 
			By Corollary \ref{coor:1}, we have 
			\[
			\Pr{X_O\le \lfloor \frac{n}{2} \rfloor}\ge \Phi\left(\frac{\lfloor \frac{n}{2} \rfloor-\Ex{X_O}}{\sqrt{\Var{X_O}}}\right)-\varepsilon\
			\]
			Note that according to the definition of $\G^+(\varepsilon)$, it holds for every perfect matchings $M$ considered in Lines~\ref{ln:for-2-start}--\ref{ln:for-2-end} that $c(M)>\varepsilon^{-2}$. 
			This allows us to apply Corollary \ref{coor:1} to $M^{(2)}_\varepsilon$, too. 
			Moreover, by the fact the $\Phi(x)$ is an increasing function in $x$ and the upper bound in \eqref{lab:4},  we now have 
			\begin{align}
			\Pr{X_{M^{(2)}_\varepsilon} \le \lfloor \frac{n}{2} \rfloor}
			&\le \Phi\left(\frac{\lfloor \frac{n}{2} \rfloor-\Ex{X_{M^{(2)}_\varepsilon}}}{\sqrt{\Var{X_{M^{(2)}_\varepsilon}}}}\right)+\varepsilon \nonumber\\
			&\le \Phi\left(\frac{\lfloor \frac{n}{2} \rfloor-\Ex{X_O}}{\sqrt{\Var{X_O}}}+3\varepsilon\right)
			+\varepsilon \nonumber\\
			&\le \Phi\left(\frac{\lfloor \frac{n}{2} \rfloor-\Ex{X_O}}{\sqrt{\Var{X_O}}}\right)+2.5\varepsilon,
			\end{align}
			where the last inequality follows from the fact that 
			$\Phi(x+h)\le \Phi(x)+\frac{h}{2}$, for every $h>0$. 
			Subtracting the above two upper bounds gives
			\[
			0\le \Pr{X_{M^{(2)}_\varepsilon}\le \lfloor \frac{n}{2} \rfloor}-\Pr{X_O\le \lfloor \frac{n}{2} \rfloor}\le 3.5\varepsilon \le  \varepsilon.
			\]
			Thus, in this case, $M^{(2)}_\varepsilon$ is an $\varepsilon$-approximation for $O$.

			%\item  
			\paragraph{Case 2.} $w(M^*_j)<\lfloor \frac{n}{2} \rfloor$. 
			According to \eqref{lab:8}, we have $w(O)-2\le  w(M^*_j)\le \lfloor \frac{n}{2} \rfloor$. 
			Thus, 
			\begin{align}\label{lab:nice}
			w(O)\le \lfloor \frac{n}{2} \rfloor+2.
			\end{align}
			In this case Line~\ref{ln:Mi-st-2} will be executed and $M^*_j$ will be set to a solution to $I_r(H^*,w,c,x_{j-1})$.  
			Let us refer to this solution as $M^{**}_j$ and further consider two cases: 
			(1) $w(M^{**}_j)\le \lfloor \frac{n}{2} \rfloor$ and 
			(2) $w(M^{**}_j)> \lfloor \frac{n}{2} \rfloor$.
			
			If $w(M^{**}_j)\le \lfloor \frac{n}{2} \rfloor$. Since $x_{j-1}\le c(O)\le x_{j}$
			\[
			c(M^{**}_j)\ge x_{j-1}\ge c(O)-\frac{1}{n}.
			\] 
			
			Moreover, by Lemma \ref{lem:reward} we have that $w(O)-2\le w(B)-2\le w(M^{**}_j)$, where $B$ is an optimal solution to 
			$I_r(H^*,w,c,x_{j-1})$. 
			Thus, similar to \eqref{lab:5}, we have
			\begin{align*}
			&\frac{\lfloor \frac{n}{2}\rfloor-w(M^{**}_j)}{\sqrt{c(M^{**}_j)}}\le\frac{\lfloor \frac{n}{2}\rfloor-w(M^{**}_j)}{\sqrt{c(O)-\frac{1}{n}}}= \frac{\lfloor \frac{n}{2}\rfloor-w(M^{**}_j)}{\sqrt{1-\frac{1}{(n\cdot c(O))}}\sqrt{c(O)}} \\
			&\le \left(1+\frac{1}{n\cdot c(O)}\right)\frac{\lfloor \frac{n}{2}\rfloor -w(M^{**}_j)}{\sqrt{c(O)}} & \text{(as $(1-x)^{-\frac{1}{2}}\le (1+x)$)} \\
			&\le \frac{\lfloor \frac{n}{2} \rfloor-w(O)}{\sqrt{c(O)}}+ \frac{2}{\sqrt{c(O)}}+ \frac{1}{n\cdot c(O)}\frac{\lfloor \frac{n}{2}\rfloor -w(M^{**}_j)}{\sqrt{c(O)}}\\
			& \leq  \frac{\lfloor \frac{n}{2} \rfloor-w(O)}{\sqrt{c(O)}}+ 3 \varepsilon &\hspace{-10mm} \text{(as $0\le w(M^{**}_j\le \frac{n}{2}$ and $c(O)\ge \varepsilon^{-2}$)}
			\end{align*}
			with which the same argument we used above shows that  $M^{(2)}_\varepsilon$ is a $\varepsilon$-approximation.
			%where the last inequality follows from $c(O)>\varepsilon^{-2}$. 
			%  Now, we have an upper bound similar to (\ref{lab:5}) and 
			%  the rest exactly follows as previous case.  Therefore, one can prove that  $M^{(2)}_\varepsilon$ is a $3\varepsilon$-approximation for the problem.  
			
			Hence, it remains to consider the case where $w(M^{**}_j)>\lfloor \frac{n}{2}\rfloor$. 
			In this case we have 
			\[
			\frac{\lfloor \frac{n}{2} \rfloor-w(M^{(2)}_\varepsilon)}{\sqrt{c(M^{(2)}_\varepsilon)}}\le \frac{\lfloor \frac{n}{2}\rfloor-w(M^{**}_j)}{\sqrt{c(M^{**}_j)}}<0,
			\]
			which implies that
			\begin{align}\label{lab:7}
			\Pr{X_{M^{(2)}_\varepsilon}\le \lfloor \frac{n}{2}\rfloor}\le 
			\Phi\left(\frac{\lfloor \frac{n}{2}\rfloor-\Ex{X_{M^{(2)}_\varepsilon}}}{\sqrt{\Var{X_{M^{(2)}_\varepsilon}}}}\right)+\varepsilon.
			\le \Phi(0)+\varepsilon
			\end{align}
			Given \eqref{lab:nice}, we have 
			\[
			\frac{\lfloor \frac{n}{2}\rfloor-w(O)}{\sqrt{c(O)}}\ge \frac{-2}{\sqrt{c(O)}}\ge -2\varepsilon.
			\]
			Hence,
			\begin{align}\label{lab:9}
			\Pr{X_O\le \lfloor \frac{n}{2}\rfloor}\ge \Phi(-2\varepsilon)-\varepsilon\ge \Phi(0)-2\varepsilon,
			\end{align}
			where we used the fact that $\Phi(x)$ is an increasing function in $x$ and $\Phi(x-h)\ge \Phi(x)- \frac{h}{2}$, for every $h>0$.
			Putting (\ref{lab:7}) and (\ref{lab:9}) together results in 
			\[
			0\le  \Pr{X_{M^{(2)}_\varepsilon}\le \lfloor \frac{n}{2}\rfloor}-\Pr{X_O\le \lfloor \frac{n}{2}\rfloor}\le 3\varepsilon \le  \varepsilon.
			\]
			Hence, we also obtain an $\varepsilon$-approximation.
			%\end{itemize}
			
			We have proven that the algorithm produces an $\varepsilon$-approximation for the \textsc{Team Order} instance. 
			Furthermore, according to Lemma~\ref{obs:1}, $|\G^+(\varepsilon)|$ and $|\G^-(\varepsilon)|$ are bounded by $n^{4\delta^{-1}\varepsilon^{-2}}$. Since the computations inside the for-loops finish in polynomial-time, the algorithm runs in polynomial-time.   
		\end{proof}

\section*{Extensions}

 Beyond the setting we have considered so far, {\sc Team Order} can also be extended in other natural ways. We discuss two extensions of {\sc Team Order}: one extends it to more than two teams, and the other considers strategic line-up choosing of both teams.
	
  	\subsection*{Team Order with More than Two Teams}
	
	A team might need to play with more than one team with the same fixed line-up of players. 
	In this case, they would like to commit to a line-up that maximizes the probability of beating a target number of teams.
	We show that even for degenerate instances, where the winning probabilities are either $0$ or $1$, computing an optimal commitment is NP-hard.  A degenerate instance involving $m$ opposing teams can be represented by a complete bipartite graph $G$ with $n$ vertices on each side.
	Suppose that our team is $T_0$, and let the opposing teams be $T_1, \dots, T_m$.
	 Vertices on the left side of $G$ represent players in $T_0$, and those on the right side represent players from the other opposing teams, whose line-ups of players are fixed.
		Moreover, every edge $\{i, j\} \in E$ is assigned with a $m$-dimensional vector, where the \gregR{$k^{\text{th}}$} component is one if player $i$ in $T_0$ beats player $j$ in team $T_k$, and it is zero otherwise.	
		The degenerated version of the \textsc{Multi-Dimensional Team Order} problem is defined as follows.
		
			\pbDef{\textsc{(Degenerate) Multi-Dimensional Team Order}}
			{A complete bipartite graph $G$ with $n$ vertices on each side, and $m$ weight functions $w_1,\dots, w_m : E \to \{0,1\}$, where $E$ is the set of edges in G. In addition, a target number $k$.}
  	 		{Does there exists a perfect matching $M$ such that $w_\ell(M) := \sum_{e \in M} w_\ell(e) \ge \floor{\frac{n}{2}} + 1$ for at least $k$ weight functions $w_\ell$, $\ell \in\{1,\dots, m\}$? }
			
	\begin{theorem}\label{thm:manyHard}
	\textsc{Multi-Dimensional Team Order} is NP-complete even when the winning probabilities are degenerate.
	\end{theorem}
	
	\begin{proof}
	The problem is in NP because given a perfect matching $M$, we can efficiently verify if $w_i(M) \ge \floor{\frac{n}{2}} + 1$ holds for at least $k$ weight functions. To show that the problem is NP-hard, we use a reduction from the \textsc{Hitting Set} problem, which is known to be NP-complete~\cite{GaJo79a}.
	
	An instance of \textsc{Hitting Set} is given by a set of elements $Q$, a collection $S$ of subsets of $Q$, and an integer $x$. 
	It is a yes-instance if there exists a size-$x$ set $q \subseteq Q$ such that $q \cap s \neq \emptyset$ for all $s\in S$.
	Let $\lambda = |Q|$.
	Without loss of generality, we can assume that $\lambda \ge x$, and both $\lambda$ and $x$ are even numbers (i.e., we can always modify an instance by adding one dummy subset or element).
	
	We now reduce a \textsc{Hitting Set} instance $(Q,S,x)$ to \textsc{Multi-Dimensional Team Order}.
	We create a bipartite graph $G$ with $\lambda$ vertices on each side; hence, $n = \lambda$. 
	We partition the indices $1,\dots, n$ of the players into three groups: 
	$N_1 = \{1, \dots, x\}$, $N_2 = \left\{x+1, \dots, x+ \frac{\lambda - x}{2} \right\}$, and $N_3 = \left\{x + \frac{\lambda - x}{2}, \dots, \lambda \right\}$.
	Let there be $m = |S|$ weight functions; hence, each weight function corresponds to a subset in $S$.
	For all $\ell = 1, \dots, m$, we define each weight function $w_\ell$ as follows.
	\begin{itemize}
	\item
	Let $w_\ell(\{i, j\}) = 1$ for all $i \in N_2$ and $j \in N$. Namely, players of $T_0$ with indices in $N_2$ win against all other players.
	
	\item
	Let $w_\ell(\{i, j\}) = 0$ for all $i \in N_3$ and $j \in N$. Namely, players of $T_0$ with indices in $N_3$ lose against all other players. 
	
	\item
	Let $w_\ell(\{i, j\}) = 1$ if the $\ell^{\text{th}}$ subset in $S$ contains the \gregR{$j^{\text{th}}$} element in $Q$.
	\end{itemize}
	Finally, we set the target number to be $k = m$. Namely, we want to decide if $T_0$ can win all the other teams by using some line-up. 
	
	Indeed, since every player in $N_2$ always wins and every player in $N_3$ always loses, it only matters how players in $N_1$ are matches to the other side of $G$. Moreover, since $|N_2| = |N_3| = \floor{\frac{n}{2}}$, to ensure that $T_0$ win every other team $T_\ell$, at least one of their player in $N_1$ must win $T_\ell$. This directly corresponds the answer to the above constructed \textsc{Multi-Dimensional Team Order} instance to that to the \textsc{Hitting Set} instance.
	
	More specifically, suppose that $(Q,S,x)$ is a yes-instance, with $q \subseteq Q$ being a size-$x$ set such that $q \cap s \neq \emptyset$ for all $s\in S$. We can connect each $i \in N_1$ to a unique vertex $j \in q$. As a result, for every $\ell = 1,\dots,m$, we have $w_\ell(\{i, j\}) = 1$ for at least one vertex $j \in q$, and $w_\ell(\{i, j\}) = 1$ for exactly $\floor{\frac{n}{2}}$ vertices $j \notin q$ (who are matched with vertices in $N_2$ and $N_3$). Hence, $w_\ell(M) \ge \floor{\frac{n}{2}} + 1$ and $T_0$ wins all other teams under this matching.
Conversely, if $(Q,S,x)$ is a no-instance, then no matter how we match vertices in $N_1$ to the other side, there will be some weight function $w_\ell$ with $w_\ell( \{i,M(i)\}) = 0$ for all $i \in N_1$, in which case we have $w_\ell(M) < \floor{\frac{n}{2}} + 1$, and $T_0$ cannot win all other teams with any line-up.
\end{proof}

        \subsection*{Strategic Behavior of Both Teams}

		    	%In what follows, we show that the optimal strategies are clear both in the case where both teams play pure strategies and when they play mixed strategies. 
			Another natural extension is the case where both teams are strategic and \gregR{decide} their line-ups simultaneously. 
			We analyze the Nash equilibrium of the resulting game. 
			Indeed, since a Nash equilibrium may not exists when players (i.e., the teams in our setting) use pure strategies (i.e., deterministic line-ups), we consider mixed strategies: each team selects a distribution of line-ups as their strategy. 
			Expected utilities are taken over both teams' line-up distributions, where the utility of a team for a pair of deterministic line-ups is their probability of winning the majority of the games under the line-ups. 
   % \greg{I think that this point requires further discussion. Why the majority and not, e.g., more than any other team?} % JG: there are only two teams now; this is not the setting with more than two teams.
It turns out that when line-ups are selected uniformly at random by both teams, a Nash equilibrium is formed.
	
		    	\begin{theorem}
		    		Each team choosing their line-ups uniformly at random is a Nash equilibrium.
		    	\end{theorem}
		    		\begin{proof}
		    	     It suffices to show that when one of the teams, say $T_1$, chooses a line-up uniformly at random, 
			      all line-ups of the other team, \gregR{say} $T_2$, yield the same utility for both teams.
			     In particular, we show that in this case, $T_2$ using the line-up $1,\dots,n$ results in the same utility as using any other line-up $\pi(1),\dots, \pi(n)$, where $\pi: \{1,\dots, n\} \mapsto \{1,\dots, n\}$ is an arbitrary permutation of $1,\dots, n$.
			     
			     Indeed, it is easy to see that if $T_1$ also permutes every line-up according to $\pi$, and selects one of the permuted line-ups uniformly at random, then $T_2$ using $\pi(1),\dots, \pi(n)$ results in a situation identical to the one where $T_2$ uses $1,\dots,n$ against the strategy of $T_1$ that is not permuted. But since $\pi$ is a one-to-one correspondence, the strategy of $T_1$ after permutation is effectively the same as \gregR{choosing} one of the line-ups uniformly at random. Hence, $\pi(1),\dots, \pi(n)$ performs equally well as $1,\dots,n$ does. This completes the proof.
			     \end{proof}

			Since the game between the two teams is essentially zero-sum, it is well-known that a Nash equilibrium is also a Stackelberg equilibrium in this case~\cite{vNeu28a}. This means that in the setting where one team has to commit to a line-up, possibly a randomized one, it is optimal to just choose a line-up uniformly at random. (If the team can only commit to a deterministic line-up, every line-up is effectively the same.)
	
		    	\begin{corollary}
		    		Choosing a line-up uniformly at random is an optimal commitment.
		    	\end{corollary}

\end{document}